\documentclass[12pt]{amsart}

\newcommand{\commentout}[1]{}

\usepackage{amsmath, amsthm, amssymb,amscd}
\usepackage{fullpage, paralist,enumerate}
\usepackage{verbatim}
\usepackage[all]{xy}
\usepackage[parfill]{parskip}
\usepackage{graphicx}

\usepackage{afterpage}

\usepackage{amsmath}
\usepackage{amsfonts}
\usepackage{mathcomp}
\usepackage{textcomp}
\usepackage{amssymb}
\usepackage{latexsym}
\usepackage{graphicx}
\usepackage{courier}
\usepackage[english]{babel}
\usepackage{mathbbol}
\usepackage{multirow}
\usepackage{latexsym}
\usepackage{epsfig}
\usepackage{subfigure}
\usepackage{dcolumn}
\usepackage{bm}
\usepackage{colordvi}
\usepackage{color}
\usepackage{booktabs}
\usepackage{stmaryrd}
\usepackage{cite}
\usepackage[linesnumbered,commentsnumbered,ruled]{algorithm2e}
\usepackage{epstopdf}
\input xy
\xyoption{all}

\newtheorem{thm}{Theorem}[section]
\newtheorem{prop}[thm]{Proposition}

\newtheorem{rmk}[thm]{Remark}

\newcommand{\nwc}{\newcommand*}

\nwc{\ben}{\begin{equation*}}
\nwc{\bea}{\begin{eqnarray}}
\nwc{\beq}{\begin{eqnarray}}
\nwc{\bean}{\begin{eqnarray*}}
\nwc{\beqn}{\begin{eqnarray*}}
\nwc{\beqast}{\begin{eqnarray*}}

\nwc{\eal}{\end{align}}
\nwc{\een}{\end{equation*}}
\nwc{\eea}{\end{eqnarray}}
\nwc{\eeq}{\end{eqnarray}}
\nwc{\eean}{\end{eqnarray*}}
\nwc{\eeqn}{\end{eqnarray*}}

\CompileMatrices

\newtheorem{theorem}{Theorem}[section]

\theoremstyle{remark}

\nwc{\nn}{\nonumber}
\nwc{\mb}{\mathbf}
\nwc{\ml}{\mathcal}

\newcommand{\lt}{\left}
\newcommand{\rt}{\right}

\nwc{\lb}{\llbracket}
\nwc{\rb}{\rrbracket}
\nwc{\bc}{{\mb c}}
\nwc{\bj}{{\mb j}}
\nwc{\mbh}{{h}}
\newcommand{\Conv}{\mathop{\scalebox{1.5}{\raisebox{-0.2ex}{$\ast$}}}}

\nwc{\vep}{\varepsilon}
\nwc{\ep}{\epsilon}
\nwc{\vrho}{\varrho}
\nwc{\orho}{\bar\varrho}
\nwc{\vpsi}{\varpsi}
\nwc{\lamb}{\lambda}
\nwc{\om}{\omega}
\nwc{\Om}{\Omega}
\nwc{\al}{\alpha}

\nwc{\IA}{\mathbb{A}} 
\nwc{\bi}{\mathbf i}
\nwc{\bo}{\mathbf o}
\nwc{\IB}{\mathbb{B}}
\nwc{\IC}{\mathbb{C}} 
\nwc{\ID}{\mathbb{D}} 
\nwc{\IM}{\mathbb{M}} 
\nwc{\IP}{\mathbb{P}} 
\nwc{\II}{\mathbb{I}} 
\nwc{\IE}{\mathbb{E}} 
\nwc{\IF}{\mathbb{F}} 
\nwc{\IG}{\mathbb{G}} 
\nwc{\IN}{\mathbb{N}} 
\nwc{\IQ}{\mathbb{Q}} 
\nwc{\IR}{\mathbb{R}} 
\nwc{\IT}{\mathbb{T}} 
\nwc{\IZ}{\mathbb{Z}} 

\nwc{\cE}{{\ml E}}
\nwc{\cP}{{\ml P}}
\nwc{\cQ}{{\ml Q}}
\nwc{\cL}{{\ml L}}
\nwc{\cX}{{\ml X}}
\nwc{\cW}{{\ml W}}
\nwc{\cZ}{{\ml Z}}
\nwc{\cR}{{\ml R}}
\nwc{\cV}{{\ml V}}
\nwc{\cT}{{\ml T}}
\nwc{\crV}{{\ml L}_{(\delta,\rho)}}
\nwc{\cC}{{\ml C}}
\nwc{\cO}{{\ml O}}
\nwc{\cA}{{\ml A}}
\nwc{\cK}{{\ml K}}
\nwc{\cB}{{\ml B}}
\nwc{\cD}{{\ml D}}
\nwc{\cF}{{\ml F}}
\nwc{\cS}{{\ml S}}
\nwc{\cM}{{\ml M}}
\nwc{\cG}{{\ml G}}
\nwc{\cH}{{\ml H}}
\nwc{\bk}{{\mb k}}
\nwc{\bn}{{\mb n}}
\nwc{\bz}{\mb z}
\nwc{\by}{\mathbf{h}}
\nwc{\bZ}{\mathbf{Z}}
\nwc{\bF}{\mathbf{F}}
\nwc{\bE}{\mathbf{E}}
\nwc{\bV}{\mathbf{V}}
\nwc{\bY}{\mathbf Y}
\nwc{\br}{\mb r}
\nwc{\pft}{\cF^{-1}_2}
\nwc{\bU}{{\mb U}}
\nwc{\bG}{{\mb G}}
\nwc{\bg}{\mathbf{g}}
\nwc{\mbf}{\mathbf{f}}
\nwc{\mbe}{\mathbf{e}}
\nwc{\be}{\mathbf{e}}
\nwc{\ind}{\operatorname{I}}
\nwc{\mbx}{\mathbf{f}}
\nwc{\bb}{\mathbf{g}}
\nwc{\xmax}{f_{\rm max}}
\nwc{\xmin}{f_{\rm min}}
\nwc{\suppx}{\hbox{\rm supp} (\mbf)}
\nwc{\cI}{\IZ^2_N}
\nwc{\chis}{{\chi^{\rm s}}}
\nwc{\chii}{{\chi^{\rm i}}}
\nwc{\pdfi}{{f^{\rm i}}}
\nwc{\pdfs}{{f^{\rm s}}}
\nwc{\pdfii}{{f_1^{\rm i}}}
\nwc{\pdfsi}{{f_1^{\rm s}}}
\nwc{\thetatil}{{\tilde\theta}}
\nwc{\red}{\color{red}}
\nwc{\blue}{\color{blue}}

\nwc{\prox}{\hbox{prox}}
\nwc{\diag}{\hbox{\rm diag}}
\nwc{\supp}{{\hbox{\rm supp}}}

\nwc{\sloc}{J_{\rm f}}
\nwc{\bu}{\xi}
\nwc{\bv}{\eta}
\nwc{\cU}{\mathcal{U}}
\nwc{\cN}{\mathcal{N}}
\nwc{\bN}{\mathbf{N}}
\nwc{\mbm}{\mathbf{m}}
\nwc{\bw}{\mathbf{w}}
\nwc{\im}{i}
\nwc{\bom}{\mathbf{w}}
\nwc{\bt}{\mathbf{t}}
\nwc{\z}{y}
\nwc{\cY}{\mathcal{Y}}
\nwc{\bM}{\mathbf{M}}
\nwc{\half}{{1\over 2}}
\nwc{\Sf}{S_{\rm f}}
\nwc{\Jf}{J_{\rm f}}
\nwc{\nul}{\hbox{\rm null}_\IR}
\nwc{\spanR}{\hbox{\rm span}_\IR}
\nwc{\Arg}{\hbox{\rm Arg~}}
\nwc{\fdr}{S_{\rm f}}
\nwc{\phase}[1]{\exp\lt[i\measured #1\rt]}
\nwc{\xnul}{x_{\rm null}}

\usepackage[utf8]{inputenc}
\usepackage{pgfplots}

\begin{document}

\title{Noise-Robust One-Bit Diffraction Tomography and Optimal Dose Fractionation}

\author{Pengwen Chen}
\thanks{Department of Mathematics, National Tsing Hua University,
Hsinchu 30013, Taiwan}
\author{Albert Fannjiang}
\thanks{Department of Mathematics, University of California, Davis, CA 95616, USA}

\maketitle
\begin{abstract}

This study presents a noise-robust framework for 1-bit diffraction tomography, a novel imaging approach that relies on intensity-only binary measurements obtained through coded apertures. The proposed reconstruction scheme leverages random matrix theory and iterative algorithms to effectively recover 3D object structures under high-noise conditions. 

A key contribution is the numerical investigation  of dose fractionation, revealing optimal performance at a signal-to-noise ratio near 1, {\em independent of the total dose}.  This finding addresses the question: How to distribute a given level of total radiation energy among different tomographic views  in order to  optimize the quality of reconstruction? 



\end{abstract}

\section{Introduction}
\begin{figure}[t]
\centering
{\includegraphics[width=14cm]{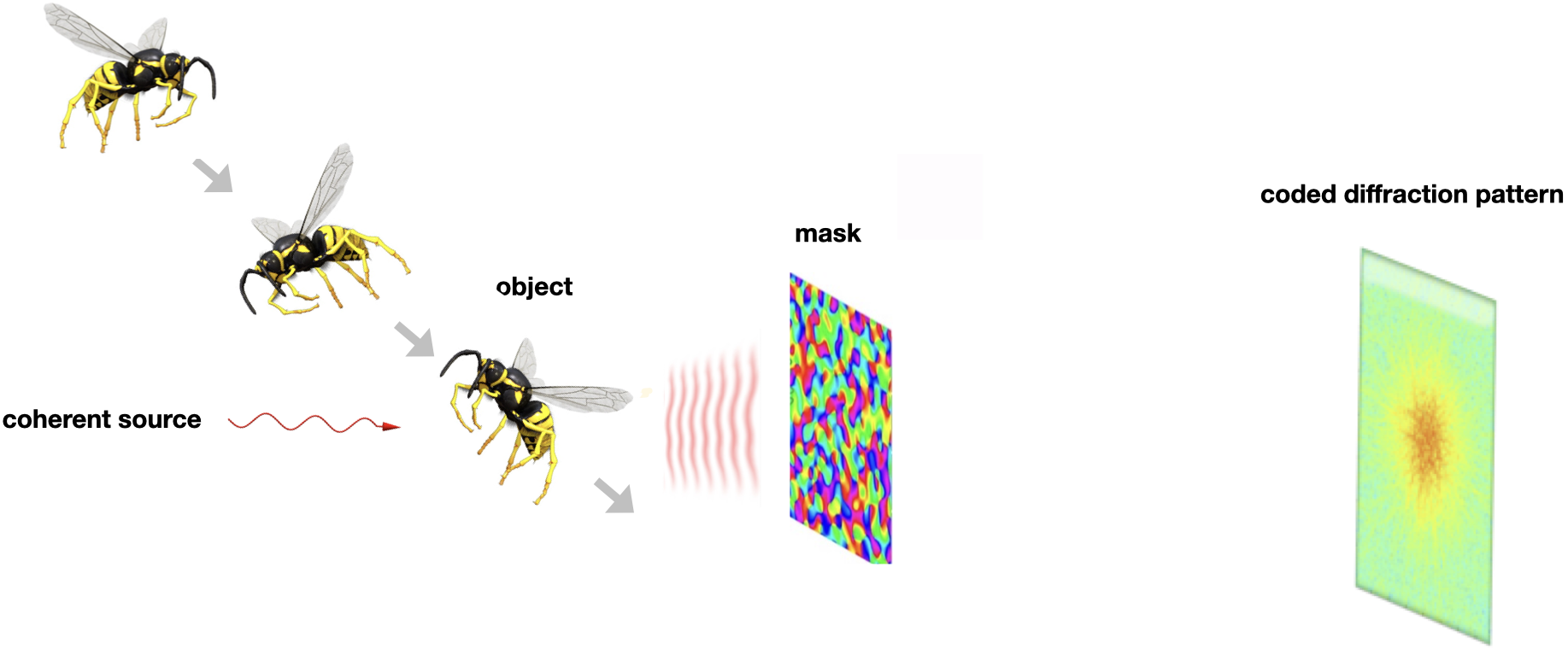}}
\caption{Coded aperture diffraction tomography:  Diffraction patterns of an object in various orientations are measured with the same random mask (see Section \ref{sec:forward})}. 
\label{fig1}
\end{figure}

Diffraction tomography is a distinct variant of tomographic imaging techniques that predominantly employs wave diffraction, as opposed to absorption, as the fundamental mode of object-wave interaction. At the heart of diffraction tomography is the goal to uncover the object's interior structure by acquiring scattered wavefield data from various orientations around the object. Unlike absorption-based methods, in this technique, both the phase and amplitude of the diffracted waves provide vital information about the object's internal structure.
This technique is particularly useful in areas such as non-destructive testing, biomedical imaging, geophysics, and more, where the wavelength of the probing wave is on par with the dimensions of the object or inhomogeneity in the medium, resulting in diffraction phenomena. 

In a traditional diffraction tomography setup, the complex-valued scattered wavefield, containing both magnitude and phase information, would be measured. However, in the case of intensity-only measurement, only the magnitude of the scattered wavefield is measured. This approach simplifies the detection process as it circumvents the challenge of phase measurement, which often requires complex and precise instruments like lenses and  interferometers, especially in high-frequency settings. 

 Further down the ladder of measurement complexity is diffraction tomography with threshold-crossing intensity-only measurements where a binary representation of the diffracted wave's intensity inherently simplifies the data acquisition and processing. A certain threshold intensity level is predefined, and the sensor only records whether the intensity of the scattered wave is above or below this threshold. 
 
This binary data is simpler to collect and process, and less sensitive to noise compared to full-waveform data. However, such thresholding inevitably leads to a loss of information about the object, making the subsequent image reconstruction process more challenging.  Traditional reconstruction techniques, like filtered backprojection or Radon inversion methods, which are suitable for high-precision measurements, may not work effectively for 1-bit intensity-only measurements due to their binary nature and the absence of phase information.

To set a reference point for the present work, let us briefly recall some key advances and insights  in the  related area of signal reconstruction of  two-dimensional  complex band-limited signals from  {\em threshold crossings in the real and imaginary parts}.  
On the one hand,
a band-limited signal whose entire extension is irreducible is uniquely
determined, up to some constant factor, by the sign information (real zero-crossings) of its real and imaginary parts, requiring essentially 2-bit information of the samples. On the other hand, this reconstruction requires extreme accuracy in identifying zero crossings \cite{Curtis87,zero89,zero90}. Furthermore, the choice of the threshold level can significantly impact the quality of the reconstruction, making it a crucial factor to consider in the design of such systems.

Therefore, in view of this instability in signal reconstruction from {\em approximate} information of zero-crossings,  the  task of tomographic phase retrieval with 1-bit threshold crossing  is quite untenable. 

A key component of our approach to mitigating the obstacle is a  randomly coded aperture \cite{unique} resulting   in coded diffraction patterns of greater diversity. The idea is motived by a Gaussian random matrix theorem  proved in \cite{null}.

The second  problem addressed in the present work is about optimal dose fractionation. 

 According to Henderson \cite{Henderson}, ``Radiation damage is the main problem which prevents the determination of the
structure of a single biological macromolecule at atomic resolution using any kind
of microscopy. This is true whether neutrons, electrons or X-rays are used as the
illumination." 

In other words, when it comes to biomedical and single molecule imaging, the resolution is damage-limited, instead of diffraction-limited. 
 On the other hand, the \emph{dose–fractionation theorem} of Hegerl--Hoppe \cite{Hegerl} set the floor for the minimum required dose and motivates the idea of dose fractionation to overcome the damage limit on resolution. 
 
In particular, the minimum dose $D$  for resolution $d$ satisfies  the dose-resolution scaling
laws -- $D\!\propto\!d^{-3}$ for incoherent and $D\!\propto\!d^{-4}$ for
coherent elastic scattering,\cite{dose-scaling04, dose-scaling09}.  Both results assume that every illumination view (or
“tilt”) is acquired with the \emph{same} fluence and that each detector
record retains full photon counts.

Despite these advances, a 
clear guideline for the best  strategy of dose fractionation is yet to emerge. 
The dose-resolution scaling laws dictate \emph{how much} dose is needed
to reach a given voxel size, but remain silent on \emph{how to distribute}
that dose across projections, especially  when the data are severely quantized:
\begin{quote}
\it Given a fixed dose budget $D$, how should one split it into
$m$ projection views (each with mean fluence $s=D/m$) so that a
\textbf{non‑linear}, 1-bit reconstruction optimizes  image quality?
\end{quote}

Because one‑bit measurements discard amplitude information above (or
below) a threshold, the marginal information gained from extra photons
in any single detector pixel rapidly saturates.  As we show below, that
saturation shifts the optimum towards many low‑dose views rather than a
few high‑dose ones, and the balance point occurs at
$
\mathrm{SNR}\;\approx\;1,
$
 independent of the total dose.

\subsection{Plan and contribution}

First, we discuss the random matrix theorem motivating the present work and the noise reduction mechanism in Section \ref{sec2}. 

We then describe the forward scattering model in Section \ref{sec:forward} and its discrete version in Appendix A.

Next in Section \ref{sec3} we discuss the power method and the shifted inverse power method for reconstruction with the discrete framework of tomography which is amenable to information-theoretical analysis as well as exact simulations 
\cite{Born-tomo,3D-phasing} (see Appendix  \ref{sec:discrete}). 

In Section \ref{sec:nsr}, we define the Noise-to-Signal ratio (NSR) for 3D tomographic phase retrieval with the Poisson noise.
 
In Section \ref{sec:num1} we discuss our tomographic sampling scheme and, in particular, the selection of threshold for noisy data.  We show that  with a relatively large number of diffraction patterns the power method is more 
advantageous while with a small to moderate number of diffraction patterns,
the shifted inverse power method converges much faster and is more stable.  

In Section \ref{sec:num2}, we  conduct simulation with dose fractionation, culminating in the numerical evidence for the optimal  dose fractionation characterized by SNR=1. 

Finally, in Section \ref{sec:final}, we discuss the practical significance of using 2-phase and 4-phase random masks and propose future research into the relation between optimal fractionation of dose and the nature of noise. 
\section{Random-matrix theorem}\label{sec2}

Consider the
nonlinear signal model: $b=|\cA f_*|$, where  $ \cA\in \IC^{M\times N}$ is the measurement matrix and $|\cdot|$ denotes entrywise
modulus. We select a threshold  to separate the ``weak"
 signals, due to destructive interference,  from the ``strong" signals,  due to  constructive interference, as follows. 
Let $I\subset \{1,\cdots,N\}$ be the support set of the weak signals (to be determined) and $I_c$ its complement such that $b[i]\leq b[j]$ for all $i\in I, j\in I_c$. Denote  the sub-row  matrices  with row indices in $I$ and $ I_c$  by $\cA_I$ and $\cA_{I_c}$, respectively.

The significance of the weak signal support $I$ lies in the fact that $I$ contains the
 loci  least sensitive to the phase information. This motivates the least squares problem: 
 \beq
\label{nul3}
\min\lt\{\|\cA_I f\|^2: f\in \IC^N, {\|f\|=\|f_*\|}\rt\}. 
 \eeq
 
A slightly simplified version of the theorem proved in \cite{null} is the following. 
\begin{theorem}\label{thm0} \cite{null}
Let $\cA$ be an $M\times N$ i.i.d. complex Gaussian matrix and $f_{\min}$ a minimizer of \eqref{nul3}.  Suppose  
\begin{eqnarray}
N< |I|\ll M\ll |I|^2.\label{scaling}
\end{eqnarray}
 Then  with an overwhelming  probability,  the relative error bound 
\beq
\label{error}
\|f_*\overline{f_*^T} -f_{\min}\overline{f^T_{\min}}\|_{\rm F}/\|f_*\|^2
&\le& c_0\sqrt{{|I|/M}} \ll 1
 \eeq
holds for some constant $c_0$, where $\|\cdot\|_{\rm F}$ denotes the Frobenius norm.
Here and below, the over-line notation denotes
the complex conjugation.

\end{theorem}

\begin{rmk}
In physical terms, the random matrix $\cA$ represent diverse, independent measurements. Since the intensity is a quadratic quantity, the weak-signal rows of $\cA$ are the place least sensitive to phase information. 

The scaling condition \eqref{scaling} dictates how to properly sort the signals into the weak and strong ones: First, the weak pixels must be more numerous than the object dimension so that the problem \eqref{nul3} does not have multiple linearly independent solutions; Indeed, the number of weak pixels should be much greater than the square-root of the total number of measurement data. Finally,  a good  scaling for selecting the weak pixels according to the error scaling in \eqref{error} is: 
\beq
\label{good}
|I|\sim \sqrt{MN}= N\sqrt{L}\quad\hbox{where $L=M/N\gg 1$ is the oversampling ratio}. 
\eeq

This physical picture naturally lands itself on the problem of  tomographic phase retrieval with a randomly coded aperture where the coded aperture enhances measurement randomness and the tomography modality provides measurement diversity (Section \ref{sec:forward}). 

In a nutshell, the strategy is this: Take a bunch of 1-bit diffraction frames, keep only the “weak’’ pixels in each, and find the object vector that drives those weak pixels as close to zero as possible.  As long as we record enough frames, that simple least-squares step under the norm constraint already lands us near the true object.
\end{rmk}
In practice, it is convenient to consider the following surrogate:
 \beq
\label{nul3''}
f_{\max}:=\arg{\max\lt\{\|\omega\odot \cA f\|^2: f\in \cX, {\|f\|=\|f_*\|}\rt\}}
 \eeq
 where  $\omega$ is the indicator vector for $I_c.$
 The asymptotical equivalence between \eqref{nul3} and \eqref{nul3''}  as $M\to \infty$  can be seen as follows.
 
As $M\to\infty$, the column vectors of $\cA=[a_{ij}]$ have nearly the same norm $M^{1/2}$
 (assume unit variance for each entry $a_{ij}$) and are nearly mutually orthogonal in the sense that
\[
M^{-1}{\sum_{i=1}^M \bar a_{ij} a_{ik}}\sim M^{-1/2}\to 0,\quad j\neq k. 
\]
In other words, we can think of $M^{-1/2}\cA$ as an isometry when $M$ is much larger than $N$. 
By the isometry property \beq
 \label{isom'}
\|f\|^2= M^{-1}\|\cA_I f \|^2+M^{-1}\|\cA_{I_c} f \|^2, 
\eeq
minimizing $ \|\cA_I f \|^2$ is equivalent to maximizing $\|\cA_{I_c}f\|^2$
over $\{f:\|f\|=\|f_*\|\}$.

\subsection{Noise robustness}\label{sec:robust}

At any given noise level,  the $I$-versus-$I_c$ membership of the indices near 
the threshold are least robust to noise while the membership of the extreme
(very strong or very weak) indices are most robust to noise.

We want to show that  these robust  indices and the corresponding row vectors also play
the strongest role in the synthesis of $f_{\rm max}$ and $f_*$.

Let $a^T_j$ denote the $j$-th row  vector of $\cA$.  We can write 
\beq
\cA^* (\om \odot \cA)= \sum_{j\in I_c} \overline{a_j} a_j^T,\label{eq5}
\eeq
the sum of rank-one projections restricted to $I_c$. 

By Theorem \ref{thm0},  the leading eigenvector  of the reduced Gram  matrix \eqref{eq5} approximates 
the true object $f_*$, so we  have 
\beq
f_*\sim\sum_{j\in I_c} \overline{a_j} (a_j^T f_*)\label{eq6}
\eeq
which can be interpreted as a linear decomposition of $f_*$ into its features
$\{\overline{a_j}:  j\in I_c\}$ with coefficients $a_j^T f_*$. As expected, the larger 
the {noiseless} data  $|a_j^T f_*|$, the more significant the corresponding feature $\overline{a_j}$ in the synthesis of $f_*$. Therein lies the de-noising mechanism of the one-bit scheme.

Before moving unto  the scattering model, let us briefly comment on how this paper is related to the literature on 1-bit compressed sensing such as \cite{1-bit-cs1,1-bit-cs2}.

While  Ref. \cite{1-bit-cs1,1-bit-cs2} assume thresholded linear measurement represented by random matrices followed by extracting sign information,  we are interested in thresholded nonlinear phase-less measurement for which the selection of threshold is a key. Indeed, our focus will be  to establish the numerical feasibility of tomographic phase retrieval with structured, properly thresholded 1-bit measurements described below.

\section{Forward scattering model}\label{sec:forward}

To facilitate  the implications of Theorem \ref{thm0}, we introduce a randomly coded aperture (i.e. mask) behind the object as shown in Figure \ref{fig1}, resulting in the following forward model. 

Suppose that  $z$ is  the optical axis in which the incident plane wave $e^{\im \kappa z}$ of wavenumber $\kappa$ propagates. 
 For a quasi-monochromatic wave field $u$ such as coherent X-rays and electron waves, it is useful to write  $u=e^{\im\kappa z} v$ to factor out  the incident wave and focus on the modulation field (i.e. the envelope), described by $v$.
 
 In view of Figure \ref{fig1},  we  now break  up the forward model into two components: First, a large Fresnel number regime from the entrance pupil to  the exit plane; Second, a small Fresnel number regime  from  the exit plane  to the detector  plane. 
 
 For the first component, we have
 \beq
v(\br)&=&e^{\im\kappa\psi(\br)} \label{rytov},\\
\label{wrap}\psi(x,y,z)&=&\int_{-\infty}^zf(x,y,z')dz'  
\eeq
\cite{3D-phasing}. 
The exit wave is given by $u=e^{\im\kappa z} v$  evaluated at the object's rear boundary (say,  $z=0$). At $z=+\infty$,  \eqref{wrap} is called  the {\em ray transform} (a.k.a. {\em projection}) of the object 
$f$ in the $z$ direction.

After passing through the object and the mask $\mu$ immediately behind, the exit wave  \eqref{rytov}-\eqref{wrap}
becomes  the masked exit wave $\mu e^{\im  \kappa \psi}$ at the exit plane $z=0$ and then undergoes Fraunhofer diffraction for $z>0$. 
The intensities of $u$ are then given by 
 \beq
|u(x,y,L)|^2=|v(x,y,L)|^2\sim |\cF[\mu \odot e^{\im  \kappa \psi} ]|^2.\label{data}
\eeq

To avoid the phase unwrapping problem associated with \eqref{rytov} we  apply the {\em weak-phase-object approximation}  (a.k.a. first-order Born approximation)
\beq
\label{Born}
e^{\im\kappa\psi}&\approx &1+\im\kappa\psi.
\eeq
At the second state, the scattered component  $v_B:=\im\kappa\psi$ is 
 first multiplied by the mask function  $\mu$  and then propagates into the far-field as $\cF (\mu \odot  v_B)$ where $\cF$ is the Fourier transform in the transverse variables.  This is the dark-field  imaging modal \cite{FPT,dark12,dark15}. The resulting coded diffraction pattern is  $|\cF (\mu \odot v_B)|^2$.

The next key ingredient to an information-based approach is discretization. 
The discrete version of forward scattering model is given in Appendix A.

\section{Power iteration and inverse power iteration}\label{sec3}

Introducing the projection $\cP_\cX$ unto the object space $\cX$, we can rewrite \eqref{nul3''} as  \beq
\label{nul3'}
\arg\max\lt\{\langle g, \cP_\cX \cA^* \diag(\omega)\cA \cP_\cX g\rangle: g\in \IC^N, {\|g\|=\|f_*\|}\rt\}
 \eeq
where we have used the fact that $\omega^2=\omega$. Here we assume 
that $\cX$ is a linear vector space resulting from, e.g. a more restrictive support constraint due to zero-padding in the mathematical set-up. 

Viewed as the Rayleigh quotient for the leading eigenvector(s) of the positive semidefinite matrix
 $\cP_\cX \cA^* \diag(\omega)\cA \cP_\cX$, \eqref{nul3'} suggests  
 the power iteration for solution 
 \beq
 \label{null-algorithm}
g^{(k+1)}= \lt(
\cP_\cX
\cA^* \diag(\omega)\cA \cP_\cX\rt)^k  g^{(1)}\|f_*\|/\| \lt(\cP_\cX \cA^* \diag(\omega)\cA \cP_\cX\rt)^k  g^{(1)}\|,
  \eeq
  for $k\in \IN$. 

Replace  the constraint $\|f\|=\| f_*\|$  in \eqref{nul3''}  with $\| \cA f\|=\|b\|$.
An alternative formulation can be developed in terms of  the  transform  domain variables $z$ as follows.
  Let $\cA^\dagger$ be  the pseudo-inverse of $\cA$ in the  object space $\cX$ 
and  $\cP=\cA\cA^\dagger $ the orthogonal projection onto  the space $\cA \cX.$

Replacing $\cA f$ by $z$ and $f$ by $\cA^\dagger z$,  we can formulate \eqref{nul3''} as
the following optimization problem
\beq\label{fourier-opt}
\lefteqn{\arg\max\lt\{\|\diag(\omega) z\|^2, z\in \cA\cX,  \|z\|=\|b\|\rt\}}\\
&=& \arg\max\lt\{ \|\diag(\omega)\cP z\|^2: z\in \IC^M, {\|z\|=\|b\|}\rt\}\nn\\
&=&\arg\max_{\|z\|=\|b\|} 
\langle z, \cP_\om z \rangle ,\nn\quad  \cP_\omega:=\cP\diag(\omega)\cP\nn
\eeq
which is the Rayleigh quotient for $\cP_\om$'s leading eigenvector. This leads to the power iteration, 
\beq\label{power} 
 z^{(k+1)}&=&\cP_\om z^{(k)}\|b\|/\| z^{(k)}\|,\quad k\in \IN.
\eeq
Using  $f^{(k)}=\cA^\dagger z^{(k)}$ and $\cA^\dagger = \cA^\dagger  \cA \cA^\dagger$, we have
\beq\label{power_f} 
 f^{(k+1)}&=& \cA^\dagger \diag(\omega) \cA  f^{(k)} \|b\|/\|   \cA  f^{(k)}  \|,\quad k\in \IN.
\eeq
See 
  Algorithm ~\ref{alg:Fourier-null}.

\begin{algorithm}
\SetKwFunction{Round}{Round}
\textbf{Input: }  The indicator vector $\omega$ for $I_c$;
\\
\textbf{Random initialization:} $\mbf^{(1)}=\mbf_{\rm rand}$, $z^{(1)}=\cA \mbf^{(1)}$.
\\
\textbf{Loop:}\\
\For{$k=1:k_{\textup{max}}-1$}
{
${\mbf^{(k+1)}}\leftarrow  \cA^\dagger (\omega \odot z^{(k) })$;\\
$z^{(k+1)}\leftarrow  \cA \mbf^{(k+1)}\|b\|/\|\cA \mbf^{(k+1)}\|$
}
{\bf Output:} {$ \mbf^{k_{\textup{max}}}$.}
\caption{\textbf{The Power Method}}
\label{alg:Fourier-null}
\end{algorithm}

\subsection{Shifted inverse power iteration} \label{sec:1.1}
The power iteration (\ref{power_f}) can be written as 
\beq\label{eq8'}
\lambda^{(k+1)} \cR^* \cR f^{(k+1)}=
\cR^* \cQ^*\diag(\omega) \cQ\cR f^{(k)},\; \lambda^{(k+1)}:=\|\cA f^{(k)}\|/\|b\|,
\eeq
which  computes the dominant  eigenvector
$\mbf_{\max}$ 
 of 
\beq
\cA^\dagger \diag(\omega) \cA=(\cR^*\cR)^{-1} \cR^* \cQ^*  \diag(\omega) \cQ\cR.
\eeq 
Hence, for some $\lambda$ we have
 \beq\label{eq224}
\lambda  \cR^* \cR \mbf_{\max}=  \cR^* \cQ^* \diag(\omega) \cQ\cR \mbf_{\max}.
\eeq
Let
 \beq
\label{AB}  \cS_{\omega} :=\cR^* \cQ^* \diag(\omega) \cQ\cR,\; \cS :=\cR^* \cQ^* \cQ  \cR=\cR^*\cR.\eeq 
Then $\lambda$ is 
 a generalized eigenvalue of the symmetric-definite pair $\{\cS_{\omega} , \cS \}$,
with the corresponding generalized eigenvector 
\beq
\mbf_{\max}= \arg\max_\mbf \frac{\langle \mbf, \cS_{\omega} \mbf \rangle}{\langle \mbf, \cS \mbf \rangle}
\eeq
where  $\mbf$ must be further restricted to the range of $\cS $ if $\cS $ is singular.
Assume that  $\cS $ is nonsingular for simplicity. 
Since both $\cS_{\omega} ,\cS $ are positive semi-definite and
 $\cS_{\omega} \preceq \cS $, then  
  eigenvalues of
  ${\cS }^{-1} \cS_{\omega} $  lie  in $[0,1)$.
   Let $\lambda_1> \lambda_2$ be the leading two eigenvalues. 
The convergence rate  of the power iteration  is given  by $\lambda_2/\lambda_1$. 

When both $\lamb_1$ and $\lamb_2$ are close to 1, the power iteration converges slowly. An effective way to speed up the convergence is to adopt the shifted inverse  power iteration related to 
 $I-{\cS }^{-1} \cS_{\omega} $. 
  Let $\mu$ be the dominant eigenvalue given by
\beq\label{eq7}
\mu=\max_{\mbf}\frac{\langle \cS \mbf, (\cS -\cS_{\omega} )^{-1} \cS \mbf
\rangle}{\langle \mbf, \cS  \mbf\rangle }
\eeq
Observe that 
\begin{eqnarray*}
 (I-{\cS }^{-1}\cS_{\omega} )^{-1}&=&(I- (\cR\cR^*)^{-1} \cR^* \cQ^* \diag(\omega) \cQ\cR)^{-1} 
\\
&=& (
(\cR^*\cR)^{-1} \cR^* (I-\cQ^* \diag(\omega) \cQ
) \cR
)^{-1}\\
&=& (\cR^* \cQ^* \diag(1-\omega )\cQ \cR)^{-1} \cR^* \cR.
\end{eqnarray*}
The shifted inverse iteration consists of generating   $\mbf^{(k+1)}$  from solving the equation,
\beq\label{eq8}
\mu^{(k+1)}(\cR^* \cQ^* (1-\diag(\omega) )\cQ \cR) \mbf^{(k+1)}= \cR^* \cR \mbf^{(k)},\; \mu^{(k+1)}\in \IR.
\eeq

\begin{algorithm}
\SetKwFunction{Round}{Round}
\textbf{Input: }  The indicator vector $\omega$ for $I_c$.
\\
\textbf{Random initialization:} $\mbf^{(1)}=\mbf_{\rm rand}$.
\\
\textbf{Loop:}\\
\For{$k=1:k_{\textup{max}}-1$}
{
Apply preconditioned CG to solve $\mbf^{(k+1)}$ from (\ref{eq8}).
}
{\bf Output:} {$ \mbf^{k_{\textup{max}}}$.}
\caption{\textbf{The Inverse Power Method}}
\label{alg:Fourier-null-inv}
\end{algorithm}

Algorithm 2  has several advantages over Algorithm 1.  
First, both algorithms converge slowly  when the spectral gap $|\lambda_1-\lambda_2|$ is close to zero but they do it at different rates. 
Observe that 
  $\lambda$ in (\ref{eq224})   and $\mu$ in \eqref{eq7} are related by 
 $\mu=(1-\lambda)^{-1}$.  Hence 
the convergence  rate of  the shifted inverse  iteration is given by 
 $( 1-\lambda_1)/( 1-\lambda_2)$, which may be significantly less than $\lambda_2/\lambda_1$ for $\lambda_1, \lambda_2$ close to $1$.

Second, when conjugate gradient methods(CG) are used to solve (\ref{eq8'}) or (\ref{eq8}), the condition number is a crucial factor in determining the convergence speed of CG. 
Let  the spectrum of $\cS^{-1} \cS_{\omega}$ lie in some interval $[\lambda_{\min},\lambda_{\max}]\subset [0,1)$. Then  the spectrum of $(I-\cS^{-1} \cS_{\omega})^{-1}$ lies in $[(1-\lambda_{\min})^{-1},(1-\lambda_{\max})^{-1}]$. Empirically  $\lambda_{\max}$ is away from $1$ and $\lambda_{\min}\approx 0$.
Hence the system (\ref{eq8}) has a much  better condition number than the system in (\ref{eq8'}), improving the performance of the  conjugate gradient method.

 In the appendices, we describe an efficient algorithm for computing the pseudo-inverse $\cR^\dagger$ and $\cA^\dagger$ 
as well as the implementation of (\ref{eq8}) 
in the context of diffraction tomography. 
The shifted inverse power itertion  is summarized in  Algorithm~\ref{alg:Fourier-null-inv}.

   \section{Noise-to-signal ratio (NSR)}\label{sec:nsr}
 
Photon noise, also called shot noise,  is due to the statistical nature of photon emission and detection. When light passes through a phase mask and creates a diffraction pattern, the photon noise in that pattern is fundamentally dictated by the number of photons detected at each point in the pattern. 

The noise level is measured by the  noise-to-signal ratio (NSR), the reciprocal of signal-to-noise ratio (SNR), given by
\beq
  \hbox{NSR}=(\hbox{SNR})^{-1}:= {\hbox{\# non-signal photons on average}\over \hbox{ \# signal photons}}. \label{55}
\eeq

Photon noise is commonly described by the Poisson distribution such that the noisy intensity data vector $\widetilde b^2$ has as the noise components  the independent Poisson random variables with the means equal to the noiseless components. To introduce the Poisson noise into
 our mathematical set-up, let $b^2=|\cA f_*|^2$  as before but consider $sb^2$ to be the noiseless intensity data with an adjustable scale factor $s>0$
 representing the strength of illumination.

  Denote the intensity fluctuation   by $z=(z_j)$. The (deterministic) noise photon count is $\sum_j  |z_j|$ 
and the total average noise photon count is given by 
  \beq\label{L1}
\sum_j  \IE |z_j|&\hbox{or more conveniently}& \sum_j \sqrt{\IE|z_j|^2}=
\Big\|\sqrt{\IE|z_j|^2}\Big\|_1,
\eeq
i.e. the root mean square of $z$, 
 where $\|\cdot\|_1$ denotes the L1-norm of the vector. 
With $z=\widetilde b^2-{s} b^2$,  the number noise photons on average is given by $\big\| \sqrt{\IE(\widetilde b^2-{s} b^2)^2}\big\|_1$.  In other words, the NSR \eqref{55} can be conveniently  written  as 
        \beq\label{nsr}
  \hbox{NSR}=(\hbox{SNR})^{-1}:= {\Big\| \sqrt{\IE(\widetilde b^2-{s} b^2)^2}\Big\|_1\over s \|b^2\|_1}= {\Big\| \sqrt{\hbox{var}(\widetilde b^2)}\Big\|_1\over  \|\IE(\widetilde b^2)\|_1}.  
   \eeq
By  a straightforward calculation with the Poisson distribution, we have
\beq
   \hbox{NSR}=(\hbox{SNR})^{-1}= {\|{b}\|_1\over \sqrt{s}\|b^2\|_1}.\label{56}
   \eeq

\section{Testing the algorithms}~\label{sec:num1}

    \begin{figure}
  \centering
\subfigure[2D image]{\includegraphics[width=4.5cm]{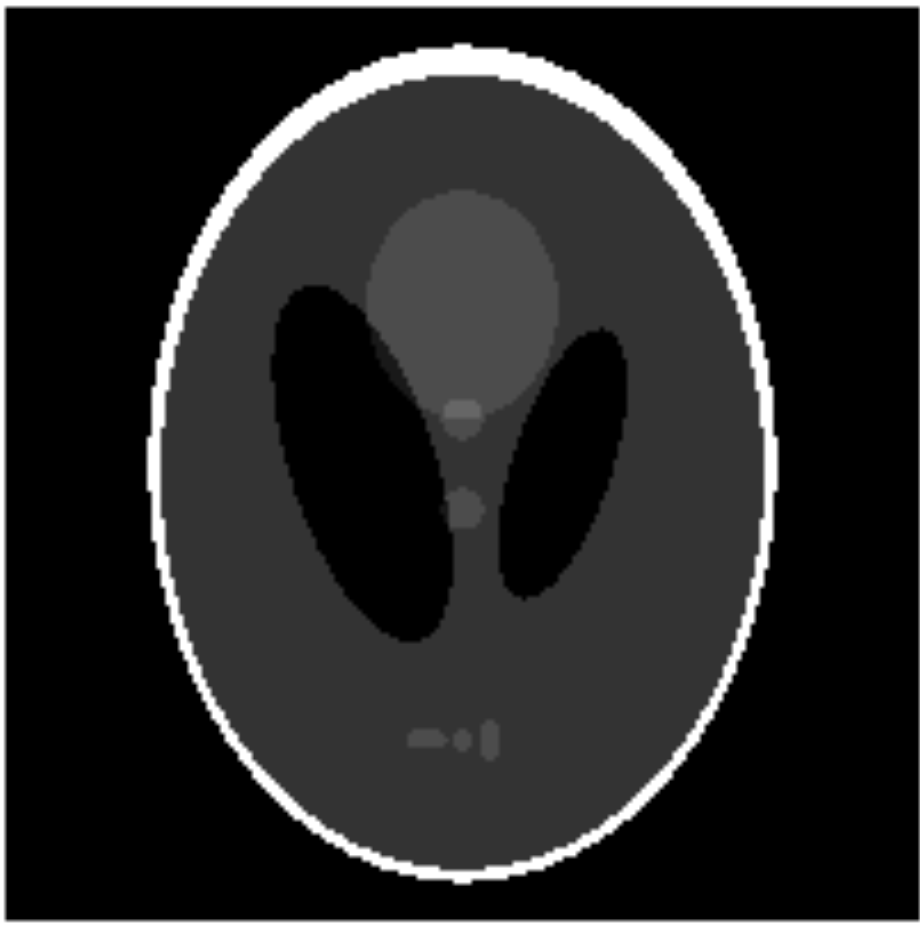}}\hspace{1cm}\hspace{1cm}
\subfigure[3D representation]{\includegraphics[width=3.5cm]{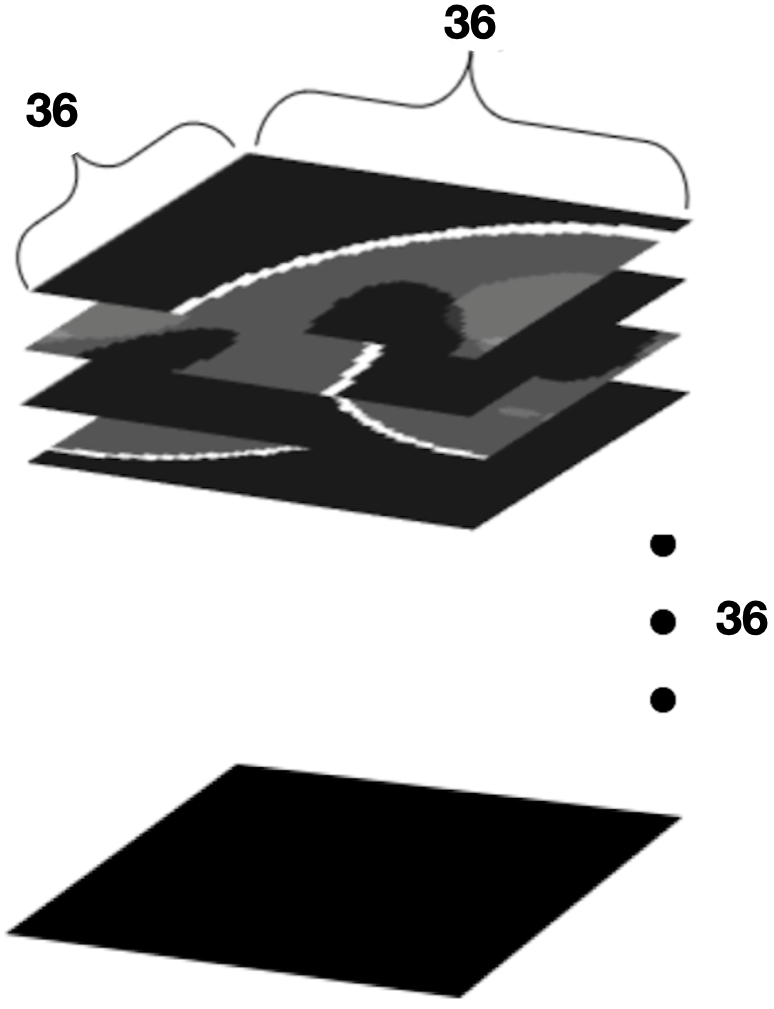}}
\caption{$216\times 216$ image  $\Longrightarrow$ $36\times 36\times 36$ object. }
\label{fig:3D}
\end{figure}

In our simulations the mask phases $ \phi$ are  independent  uniform random variables over $[0,2\pi)$. We will also test 2-phase and 4-phase random masks,  respectively exemplified by
the half-wave and quarter-wave plates for X-rays and $\pi$-shift and $\pi/2$-shift phase plates in electron microscopy.  The phases of the 4-phase mask phases are independent uniform variables  in
 $\{0, \pm \pi/2, \pi\}$ while  the 2-phase  mask phase are the Bernoulli random variables from $\{0, \pi\}$. 
Like the continuous phase random mask, the 2-phase and 4-phase mask pixels have zero mean and unit variance.

To aid visualization, we construct the {\em complex-valued} 3D object from the $216\times 216$ phantom (Fig. \ref{fig:3D} (a)) by partitioning the real-valued phantom image into 36 pieces, each of which is $36\times 36$ and stacking them into a $36\times 36\times 36$ cube (Fig. \ref{fig:3D}(b)). We then randomly modulate the phase of each voxel. The resulting 3D object is called 3D randomly phased phantom (RPP).
 We shall refer to the corresponding 2D randomly phased phantom as the {\em flattened version} of the 3D object.

For  show the results for the  absolute correlation between $f_*$ and reconstruction $f$ given by 
\[
R(f,f_*):={|f^* f_*|\over \|f\| \|f_*\|}
\]
for a pixel-wise, zero-mean object such as RPP. 
$R$ is closely related to  the structure component of  the structural similarity index measure (Appendix \ref{sec:ssim}), especially when the pixel averages are nearly zero. 

When $f$ is the leading eigenvector corresponding to $\lamb_1$, we will denote the correlation as $R_1$;
when $f$ is the second leading eigenvector corresponding to $\lamb_1$, we denote the correlation as $R_2$.

We use the reconstruction scheme \eqref{power} with the threshold selected separately
for each coded diffraction pattern $j=1,\cdots, m$ according to the guideline:
\begin{itemize}
\item For small NSR, Theorem \ref{thm0} suggests a small value of $|I_j|/p^2<\half$
\item For large NSR, we adopt the median rule,
i.e. $|I_j|/p^2=\half$
\end{itemize}
where $I_j$ is the index set of weaker signals in the $j$-th coded diffraction pattern. 
Specifically, we adopt the rule 
\beq
\label{threshold}
|I_j|/ p^2= \min(1/4+\hbox{NSR}/4,1/2)
\eeq
which  is consistent with the condition \eqref{scaling}. 

To avoid  the missing cone problem in tomography, we consider
$m=3\rho n$  more or less evenly distributed  random directions with the adjustable  parameter $\rho>0$ (see \eqref{tset} in Appendix A). 
According to \cite{Born-tomo,3D-phasing}, a non-degenerate set of $n+1$ directions is the minimum requirement for   discrete tomography with noiseless data. For 1-bit diffraction tomography with highly noisy data, however, one should deploy a larger set of directions for  any reasonable reconstruction.  

\begin{figure}
\centering
\subfigure[$\rho=1$, NSR=0.0]{  \includegraphics[width=0.3\textwidth]{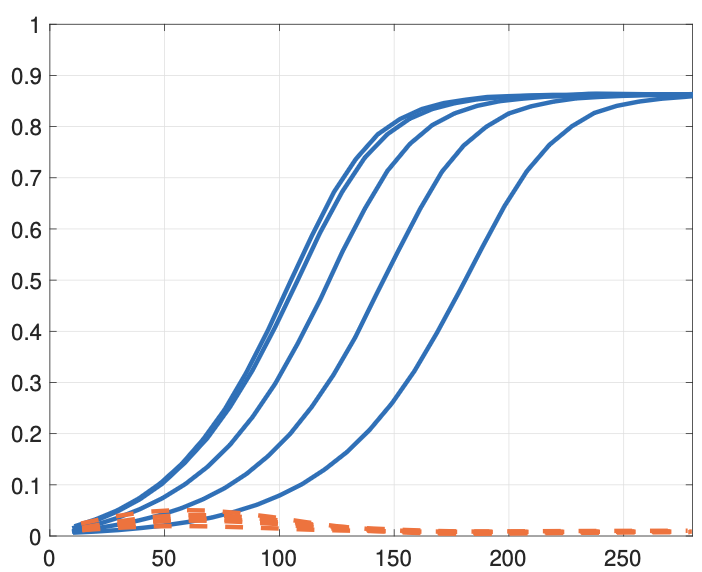}}\quad
\subfigure[$\rho=2$, NSR=0.0]{  \includegraphics[width=0.3\textwidth]{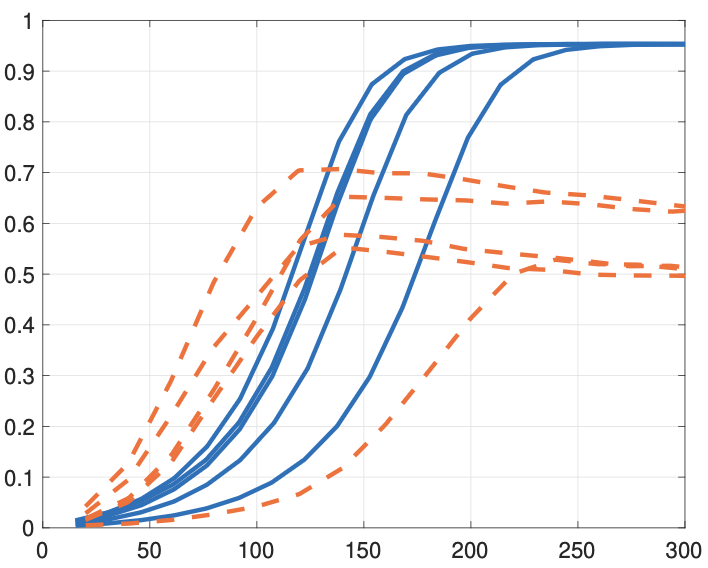}}\quad
\subfigure[$\rho=4$, NSR=0.0]{  \includegraphics[width=0.3\textwidth]{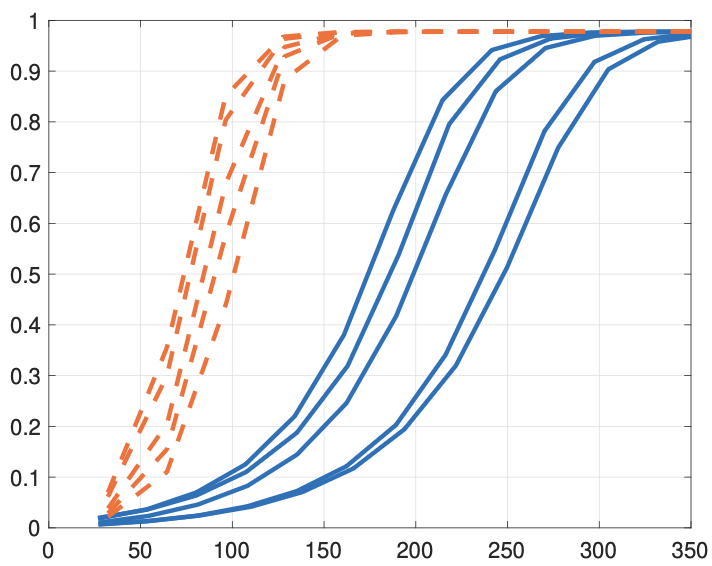}}\\
\subfigure[$\rho=1$, NSR=0.5]{  \includegraphics[width=0.3\textwidth]{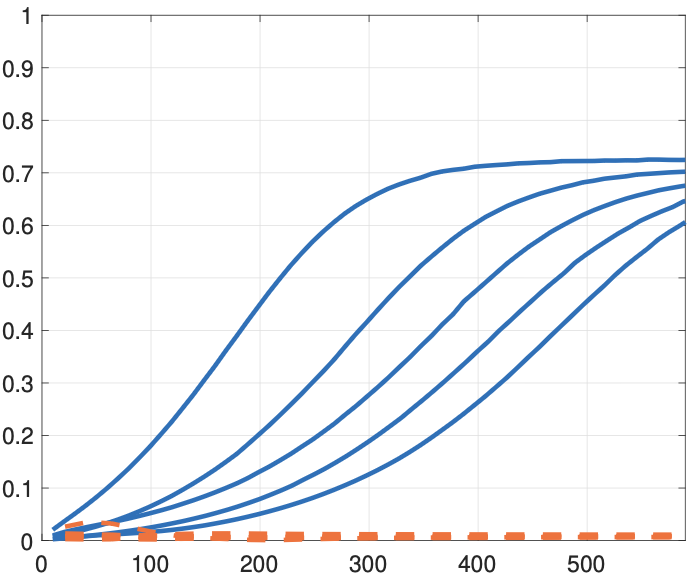}}\quad
\subfigure[$\rho=2$, NSR=0.5]{  \includegraphics[width=0.3\textwidth]{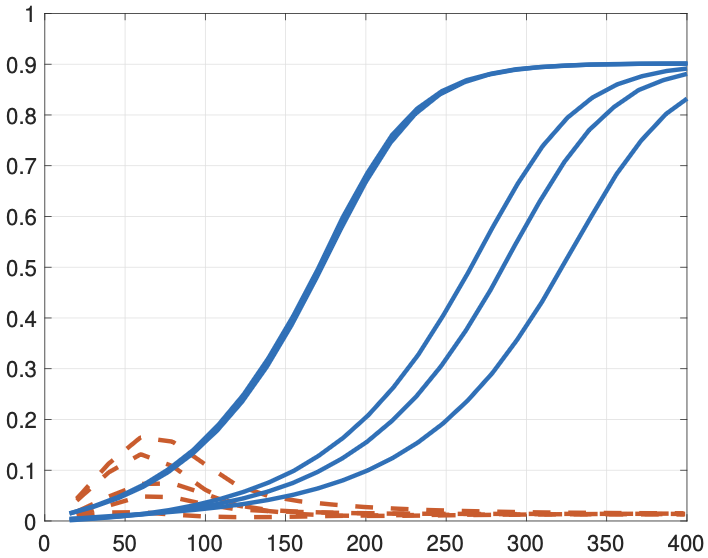}}\quad
\subfigure[$\rho=4$, NSR=0.5]{  \includegraphics[width=0.3\textwidth]{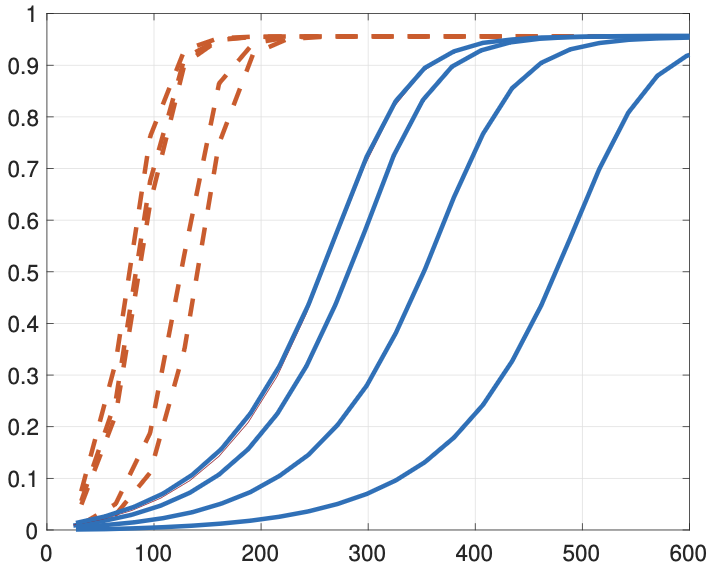}}\\
\subfigure[NSR=0.0]{  \includegraphics[width=0.45\textwidth]{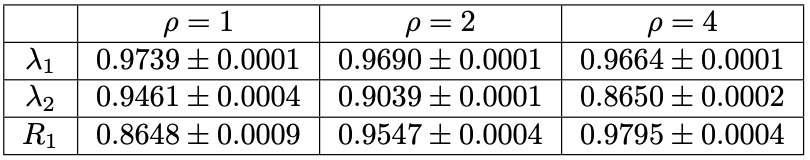}}
\subfigure[NSR=0.5]{  \includegraphics[width=0.45\textwidth]{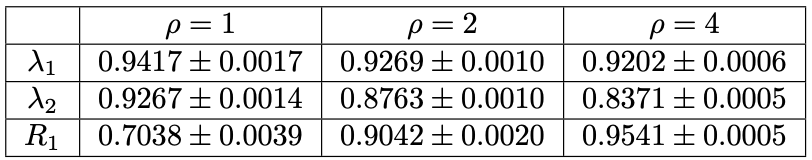}}
\caption{Correlation versus computation time (in second) with 5 independent trials with the power method (dashed line) and the inverse power method (solid line). The number of diffraction patterns is $m=3\rho n$. }
\label{fig:time_R}
\end{figure}

First let us compare the convergence rates of Algorithm 1 and Algorithm 2.  

Here and below, 5 independent trials (with independent Poisson random variables and initial points) are
conducted  to compute the mean and standard deviation of $R_1, R_2, \lamb_1,\lamb_2$ as indicated  by the error bars. In each trial, the random mask and the set of projections are kept unchanged.

Figure~\ref{fig:time_R} shows the
convergence comparison for various NSRs and $\rho$'s. We make the following observations:

\begin{itemize}
\item With a sufficiently large number of diffraction patterns ($\rho\ge 4$), the power iteration converges with faster and  more stable rate than the inverse power method. 

On the other hand, when $\rho$ decreases the inverse power method degrades in a more graceful manner than the power method. 
\item There is significant denoising effect in both methods (as NSR increases from 0 to 0.5) with a sufficient number of diffraction patterns. This is consistent with the denoising mechanism pointed out in Section \ref{sec:robust}. 

\item With a sufficient number of diffraction patterns ($\rho=4$) the convergence behavior of the power method appears less sensitive to the initial point than the inverse power method. However, the eigenvalues regulating the eventual convergence and the correlation of the reconstruction with the inverse power method are robust to the Poisson noise and the initial points, c.f. Figure \ref{fig:time_R} (g)(h).
\end{itemize}

\begin{figure}
 \centering   
\subfigure[Eigenvalues vs. $\rho$]{  \includegraphics[width=6cm]{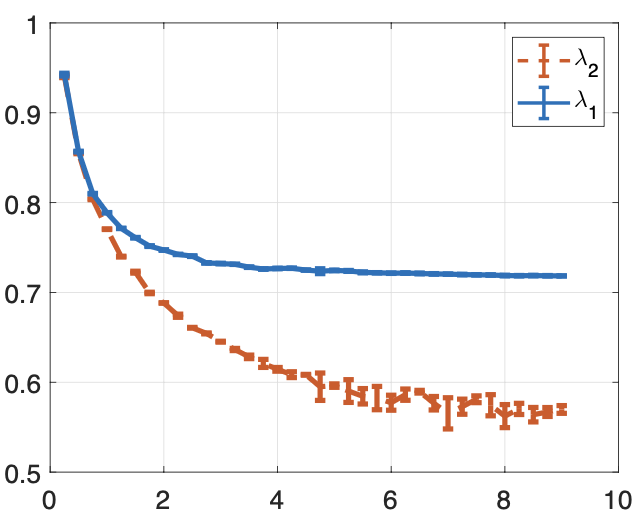}}\quad\quad
\subfigure[Correlation vs $\rho$]{  \includegraphics[width=6cm]{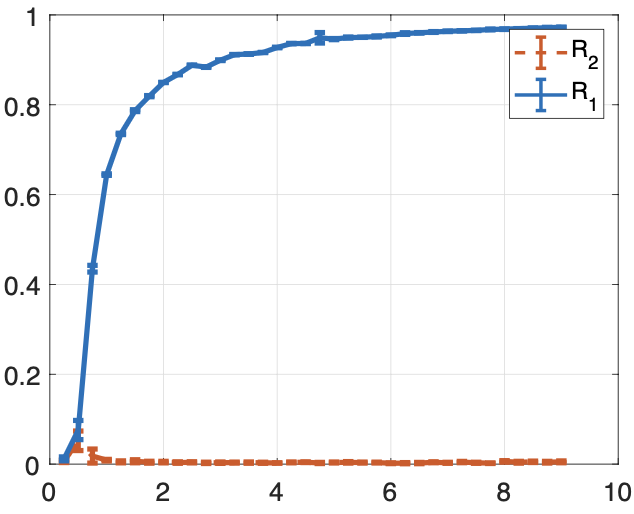}}
\caption{(a) The two leading eigenvalues and (b) the correlations as function of $\rho$ at NSR =1 where
$R_1$ and $R_2$ are respective correlations of the two leading eigenvectors with the original object.  The number of diffraction patterns is $m=3\rho n$. }
\label{fig6}
\end{figure}

Figure \ref{fig6} shows how the eigenvalues and the correlations change with the number of diffraction patterns at a larger noise level NSR = 1. We also compute  the second leading eigenvector of $\cS^{-1}\cS_\omega$ to  shed a light on the relation between  the convergence behavior and  the  spectral gap. To this end,
we employ the Krylov subspace methods, described in Appendix~\ref{sec:1.3}, to extract the second leading eigenvector. 
We observe that
\begin{itemize}
\item The spectral gap for the inverse power method initially widens with $\rho$ and then gradually saturates (Figure \ref{fig6}(a));
\item  The reconstruction correlation sharply rises  and then plateaus (and Figure \ref{fig6}(b)). 
\end{itemize}
By and large, there is little variation across independent trials, except for the second eigenvalue with large $\rho$ (Figure \ref{fig6}(a)). 

\begin{figure}
 \centering   
\subfigure[R vs NSR $\in (0,1)$]{  \includegraphics[width=0.3\textwidth]{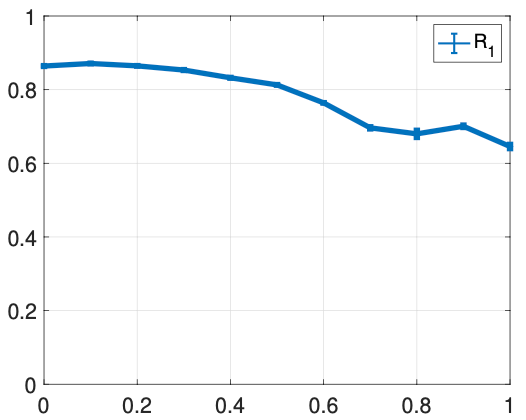}}\quad
\subfigure[NSR=0.5,$R_1$=0.8344]{  \includegraphics[width=3.7cm]{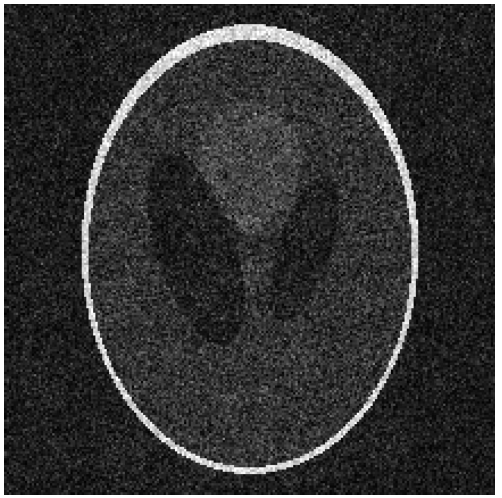}}\quad
\subfigure[NSR=1.0,$R_1$=0.6574]{  \includegraphics[width=3.7cm]{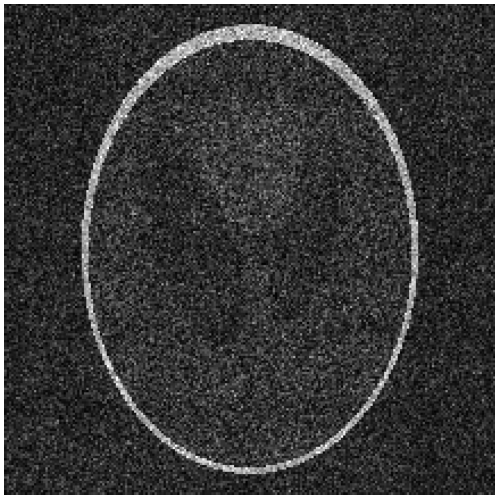}}\\
\subfigure[R vs NSR $\in (0,1)$]{  \includegraphics[width=0.3\textwidth]{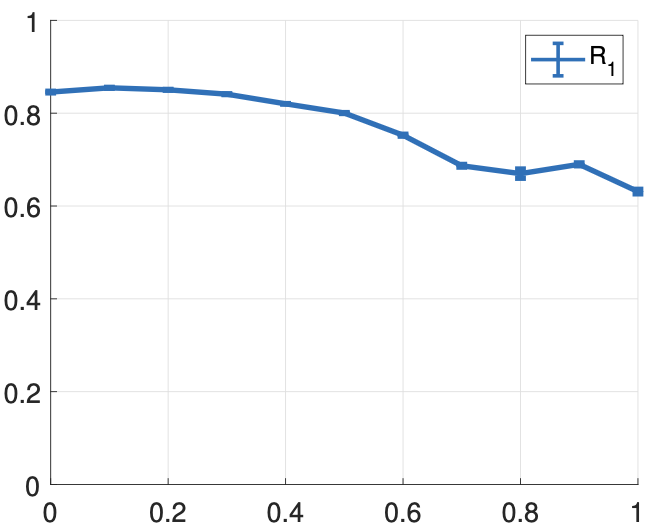}}\quad
\subfigure[NSR=0.5,$R_1$=0.8002]{  \includegraphics[width=3.7cm]{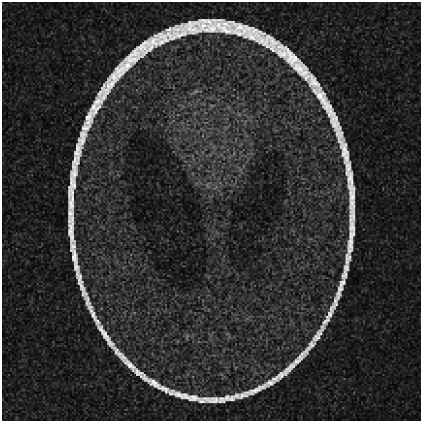}}\quad
\subfigure[NSR=1.0,$R_1$=0.6368]{  \includegraphics[width=3.7cm]{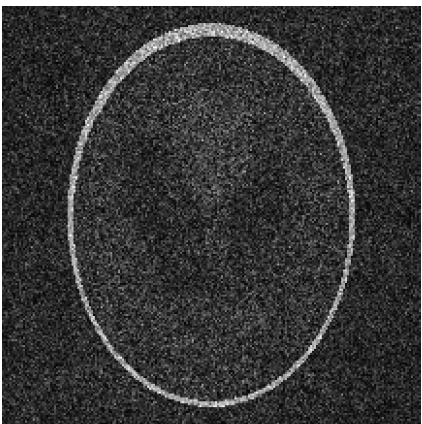}}
\caption{The correlation of the RPP reconstruction with $3n$ diffraction patterns vs NSR $\in (0,1)$ and the flattened  magnitude $|f|$ with (bottom) and without (top) a beam stop.  The beam stop covers the $1\%$ center area of each diffraction pattern.  Different direction sets $\cT$ are independently selected for different NSRs.}
\label{fig:rho1and2}
\end{figure}

\begin{figure}
\centering
\subfigure[R vs NSR $\in (0,1)$]{  \includegraphics[width=5cm]{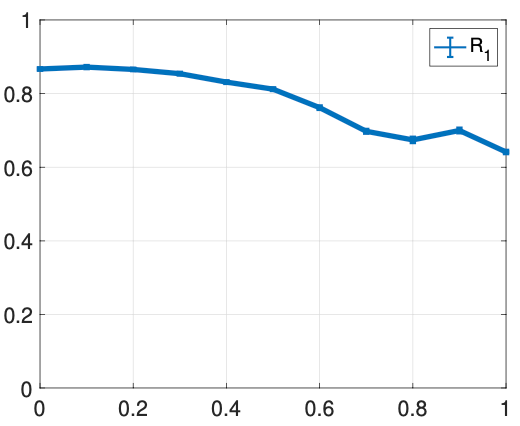}}\quad
\subfigure[NSR=0.5, R=0.8380]{  \includegraphics[width=4cm]{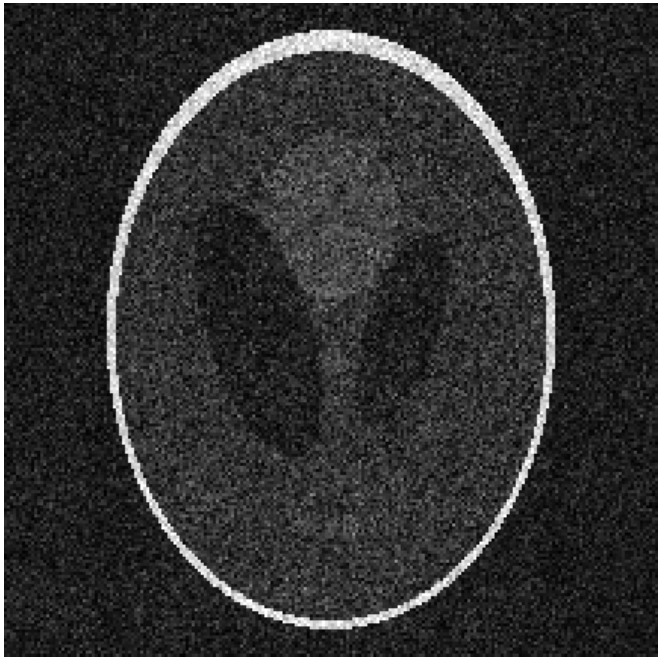}}\quad
\subfigure[NSR=1.0, R=0.6639]{  \includegraphics[width=4cm]{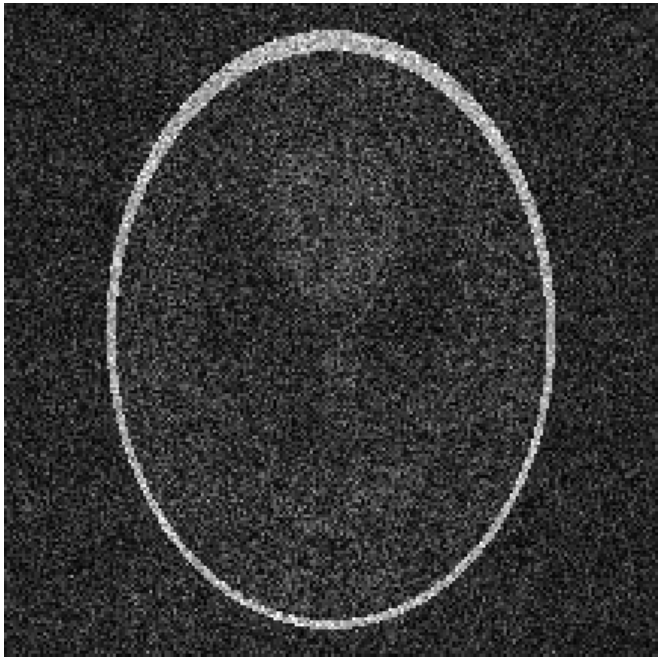}}\\
\subfigure[R vs NSR $\in (0,1)$]{  \includegraphics[width=5cm]{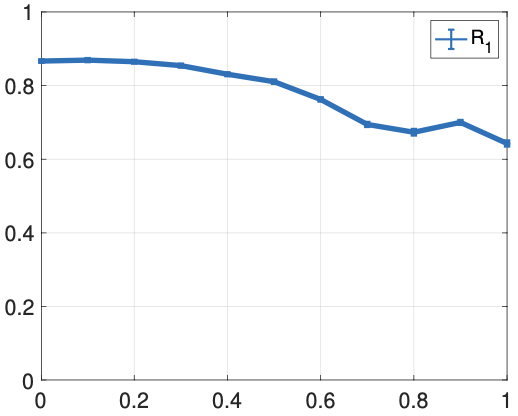}}\quad
\subfigure[NSR=0.5, R=0.8366]{  \includegraphics[width=4cm]{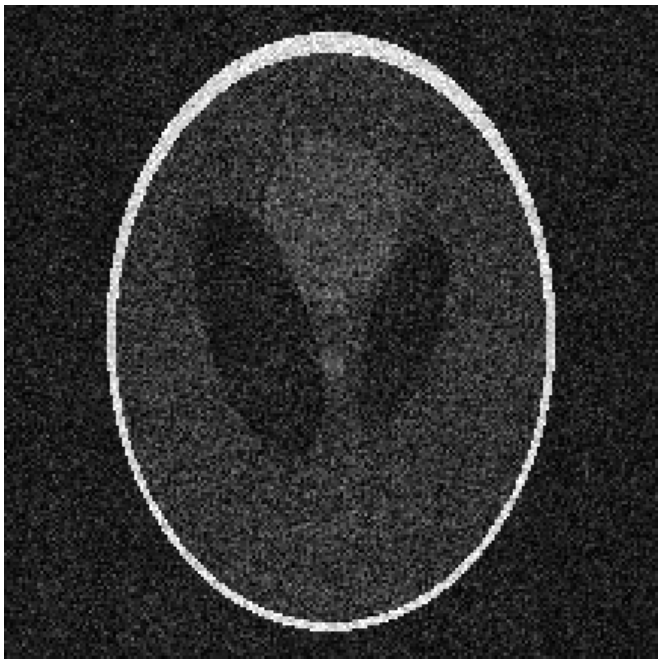}}\quad
\subfigure[NSR=1.0, R=0.6684]{  \includegraphics[width=4cm]{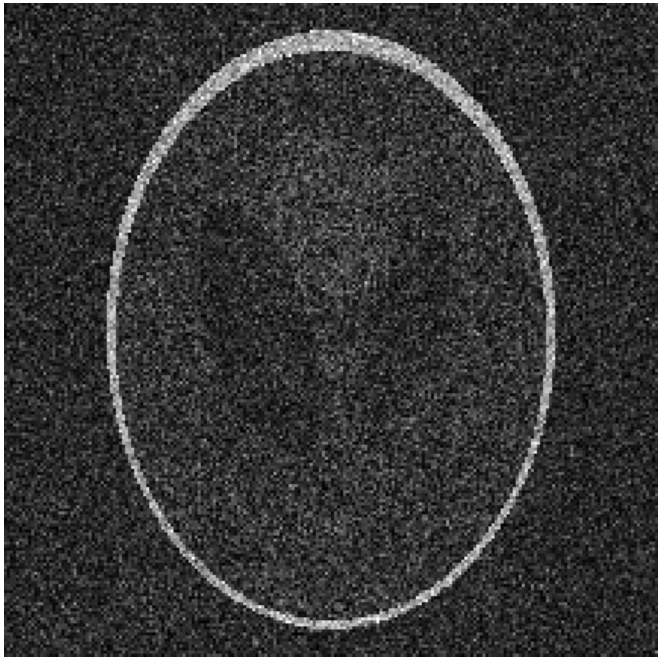}}\\
\caption{The correlation of the RPP reconstruction with $3n$ diffraction patterns vs NSR $\in (0,1)$ and the flattened  magnitude $|f|$ with a 4-phase (top) and 2-phase  (bottom) random mask.}
\label{fig:2-4phase}
\end{figure}

 In the transmission geometry of X-ray diffractive imaging,  a beam stop is usually added to prevent the high-intensity, direct (unscattered) beam from damaging the pixel sensor. 
  Figure~\ref{fig:rho1and2} compares  the reconstruction with and without a beam stop
  using $3n$ diffraction patterns ($\rho=1$ in \eqref{tset}). In modern synchrotron and XFEL beamlines, the percentage of the detector area covered by a beam stop is usually much less than $1\%$ \cite{stop1, stop2}, we simulate the effect of a beam stop by setting the central $1\%$ pixels of each diffraction pattern to zero. Figure \ref{fig:rho1and2} (b) shows moderate degradation in the quality of reconstruction ($\sim 4\%$ reduction in $R$) due to the presence of a beam stop. 
Comparing   Figure~\ref{fig:rho1and2} with Figure 8 of \cite{3D-phasing}  we find unsurprisingly  that  reconstruction with 1-bit data has a poorer quality than that with full intensity data.

Next we test the performance with a 2-phase and 4-phase random masks.
Figure \ref{fig:2-4phase}  shows that the quality of reconstruction of the 2-phase and  4-phase masks is nearly identical to that of a continuous phase mask, indicating that speckle richness comes from randomness, not fine phase quantization.
This opens the door to low-cost, high-efficiency binary (or 4-level)
phase plates and modulators without compromising reconstruction quality (see Section \ref{sec:final}).

\section{Optimal fractionation of radiation dose}\label{sec:num2}

The idea of dose fractionation emerged primarily in the context of electron microscopy, where highly sensitive samples are prone to radiation damage  \cite{Hegerl,Hoppe,dose-num95}. The core principle of dose fractionation is to split the total electron dose into multiple, lower-intensity exposures rather than applying the entire dose in a single or few exposures. This technique balances the need for sufficient signal to image fine details with the need to limit radiation damage that can degrade or destroy the sample. Dose fractionation is an operation principle in the successful method of single-particle
imaging in electron microscopy \cite{Frank06}. 

Likewise, by using a properly fractionated dose per diffraction pattern, tomographic X-ray diffractive imaging can optimize  3D structure determination  with a fixed level of risk of damaging the sample \cite{Jacobsen}.

As proved in  \cite{3D-phasing},  unique $n^3$-structure determination requires at least $n+1$ {\em noiseless} diffraction patterns. With the Poisson noise model, noiseless diffraction patterns contain an infinitely large dose. 
A finite level of total dose necessarily  results in noisy diffraction patterns whose NSR increases with  the number of diffraction patterns $m$.

As the degree of dose fractionation  is proportional to the number of diffraction patterns $m$,  it is natural to ask how  $m$ scales with NSR for a fixed $D$.  As $m$ is proportional to $D$, NSR is a key parameter for characterizing an optimal fractionation.

\subsection{Dose as the mean photon number}

\begin{figure}[t]
 \centering   
  \subfigure[$D=1.2\times 10^6$]{  \includegraphics[width=5cm]{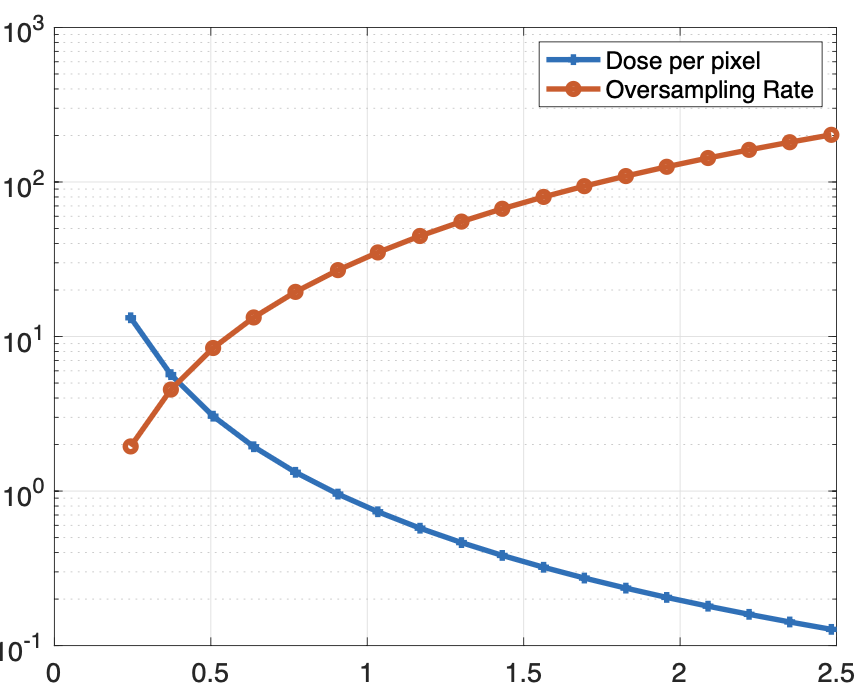}}\quad
   \subfigure[$D=2.5\times 10^6$]{  \includegraphics[width=5cm]{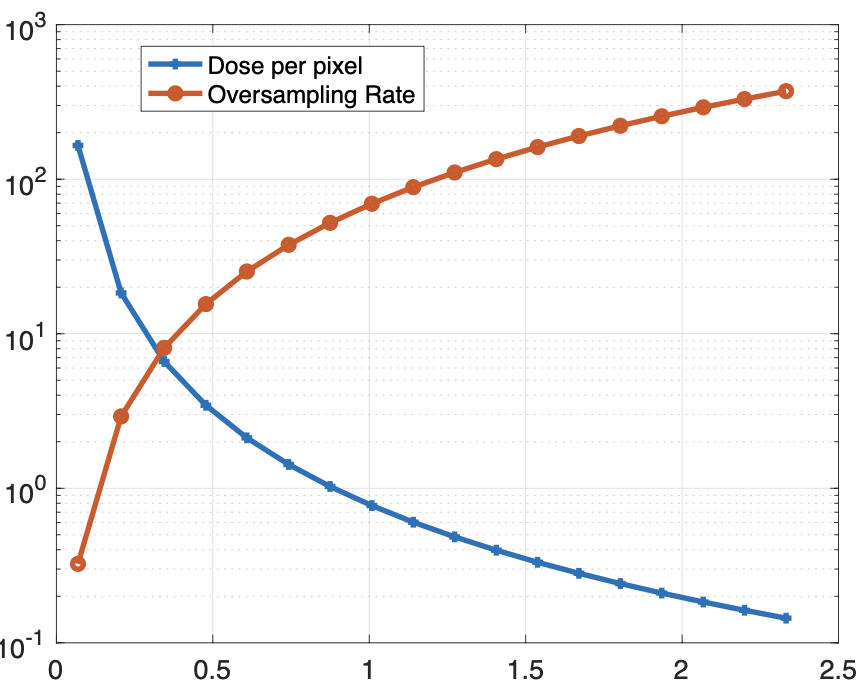}}\quad
 \subfigure[$D=5\times 10^6$]{  \includegraphics[width=5cm]{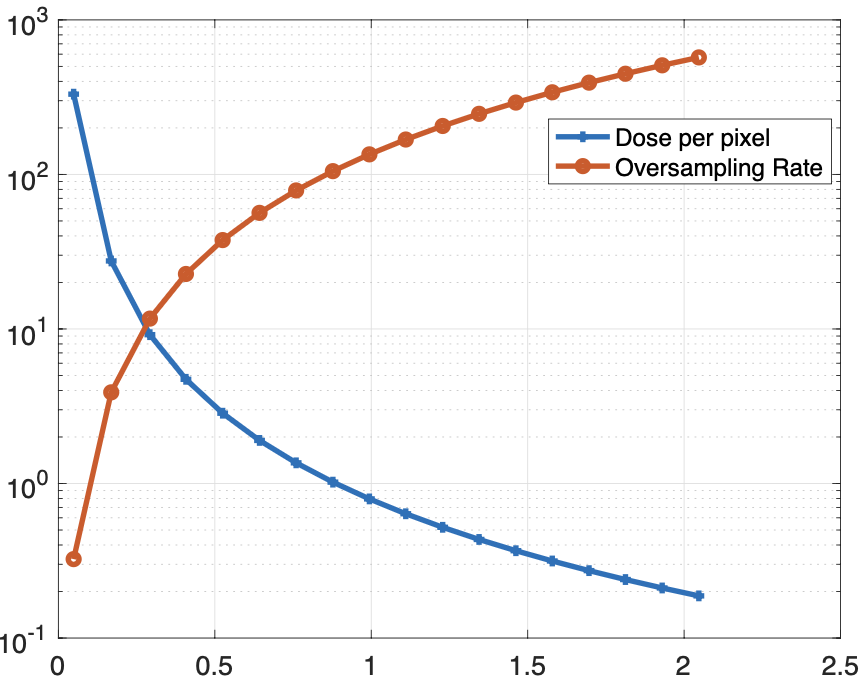}}
\caption{Dose per pixel and oversampling ratio versus NSR with various dose level as defined in \eqref{dose2}. }
\label{fig:rho-new}
\end{figure}

Let us consider the mean photon number
\beq
D&=&s\|b^2\|_1\label{dose2}
\eeq
 as the definition of  dose.

To find the dependence of  $m$  on  NSR
   for a fixed dose, rewrite \eqref{56}
 as
   \[
   \hbox{NSR}= {\sqrt{s}\|{b}\|_1/D} =  {\|{b}\|^2_1\over  \hbox{NSR}\|b^2\|_1D}. 
   \]
 Hence, 
    \beq\label{2nd'}
m\sim  \|b\|_1=(\hbox{NSR})^2\cdot D {\|b^2\|_1\over \|b\|_1}
    \eeq
whose right hand side is roughly independent of $m$ for a fixed $D$ (as $\|\cdot\|_1$ appears in both the numerator and denominator).  
In other words, the number of diffraction patterns scales quadratically with NSR for a fixed $D$. 

Figure \ref{fig:rho-new} shows the oversampling ratio $L$, defined as the ratio of the number  of measurement data $m (2n-1)^2$ to the object dimension $n^3$ ($\approx 12\rho$), and the dose per pixel as  NSR varies for three different levels of dose. 

\subsection{Alternative definition of dose}

While the mean photon number provides a straightforward measure of the energy distribution, let us consider an alternative definition of dose as  the root mean square, taking into account both its amplitude and variability:
\beq
D&:=&\Big\| \sqrt{\IE(\widetilde b^4)}\Big\|_1=\|\sqrt{s^2b^4+\IE(z^2)}\|_1
\eeq
as the dose metric. This definition  
puts greater weight on high energy levels than \eqref{dose2} does.
For the Poisson statistics,  $\IE(z^2)=sb^2$ and hence
\beq
D&=&\|\sqrt{s^2b^4+sb^2}\|_1=\sqrt{s}\|b\sqrt{sb^2+1}\|_1.\label{dose}
\eeq
To find the dependence of $m=3\rho n$  on  NSR
   for a fixed dose, we need to solve the two nonlinear equations, 
    \eqref{dose} and \eqref{56}.

   For $s$ large (small NSR), $\sqrt{s^2b^4+sb^2}\sim sb^2$ and $D\sim s\|b^2\|_1$. Combining with \eqref{56}, we have
    \beq\label{2nd}
  m\sim  \|b\|_1\sim (\hbox{NSR})^2\cdot D {\|b^2\|_1\over \|b\|_1}
    \eeq
    with the right hand side roughly independent of $m$. 
    
    For $s$ small (large NSR), $\sqrt{s^2b^4+sb^2}\sim \sqrt{s} b$ and $D\sim \sqrt{s}\|b\|_1$. 
    Combining with \eqref{56}, we have
    \beq\label{1st}
m\sim    \|b\|_1\sim  \hbox{NSR} \cdot D {\|b^2\|_1\over \|b\|_1}
    \eeq
 with the right hand side roughly independent of $m$. 
 
 Hence, instead of a single quadratic scaling law \eqref{2nd'} for the definition \eqref{dose2}, there are two different scaling behaviors (quadratic for small NSR  or linear for large NSR) in  the number of diffraction patterns as NSR varies.

\begin{figure}[t!]
 \subfigure[$R_1,R_2$ vs NSR]{  \includegraphics[width=5cm]{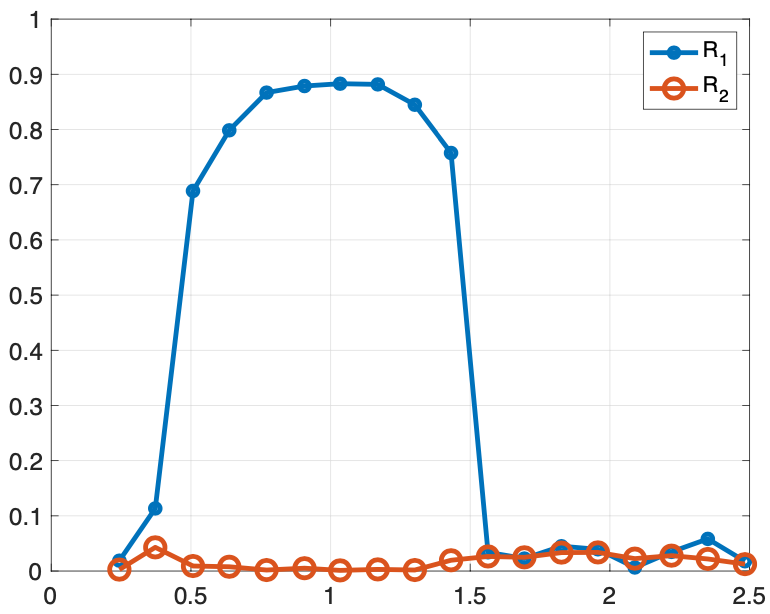}}\quad
 \subfigure[$R_1,R_2$ vs NSR]{  \includegraphics[width=5cm]{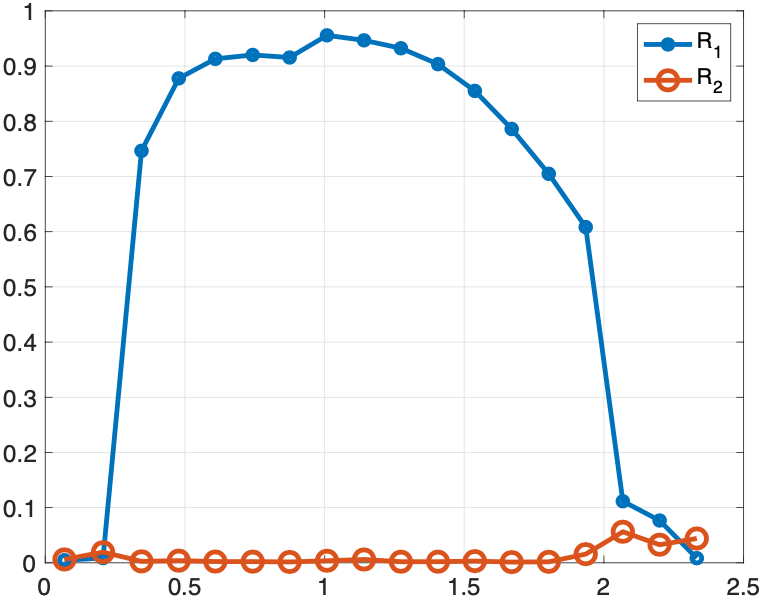}}\quad
 \subfigure[$R_1,R_2$ vs NSR]{  \includegraphics[width=5cm]{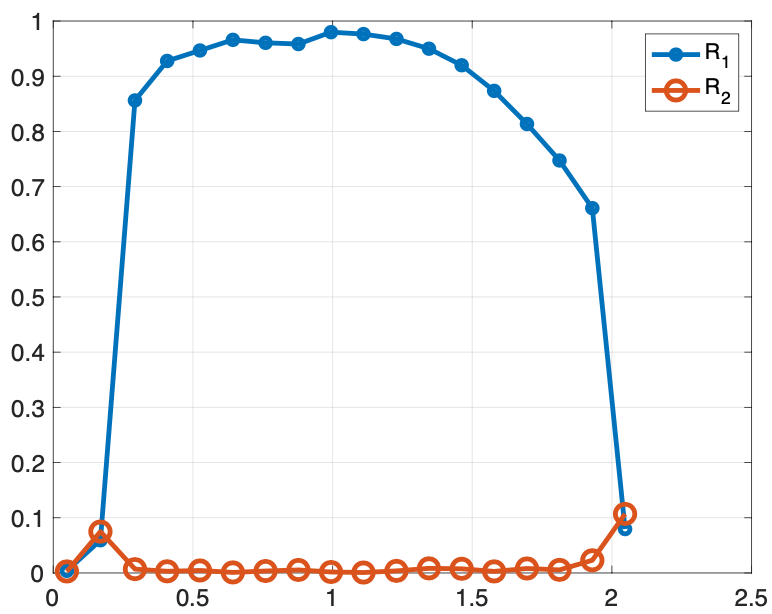}}\\
\subfigure[$\lamb_1, \lamb_2$ vs NSR]{  \includegraphics[width=5cm]{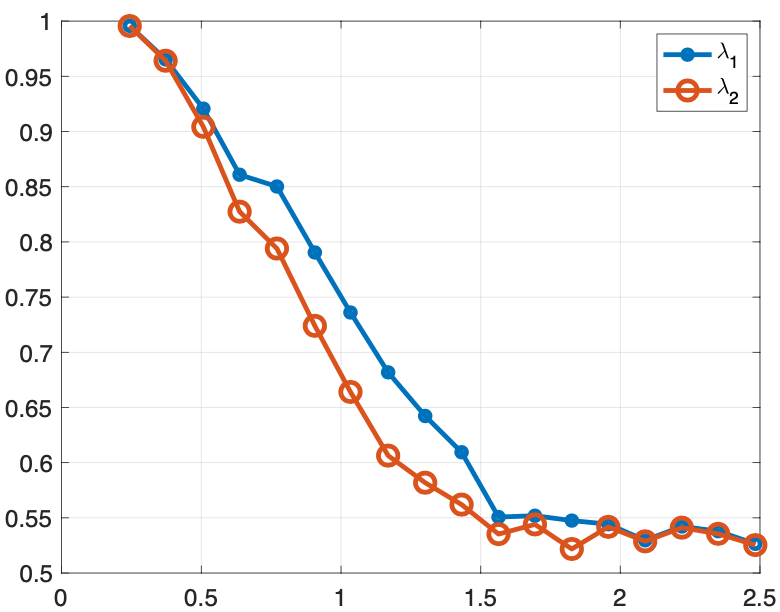}}\quad
\subfigure[$\lamb_1, \lamb_2$ vs NSR]{  \includegraphics[width=5cm]{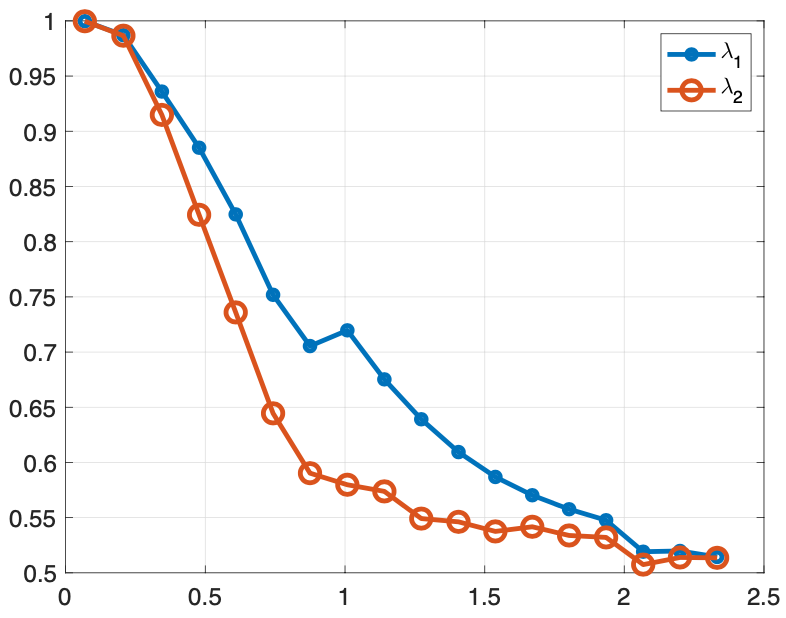}}\quad
\subfigure[$\lamb_1, \lamb_2$ vs NSR]{  \includegraphics[width=5cm]{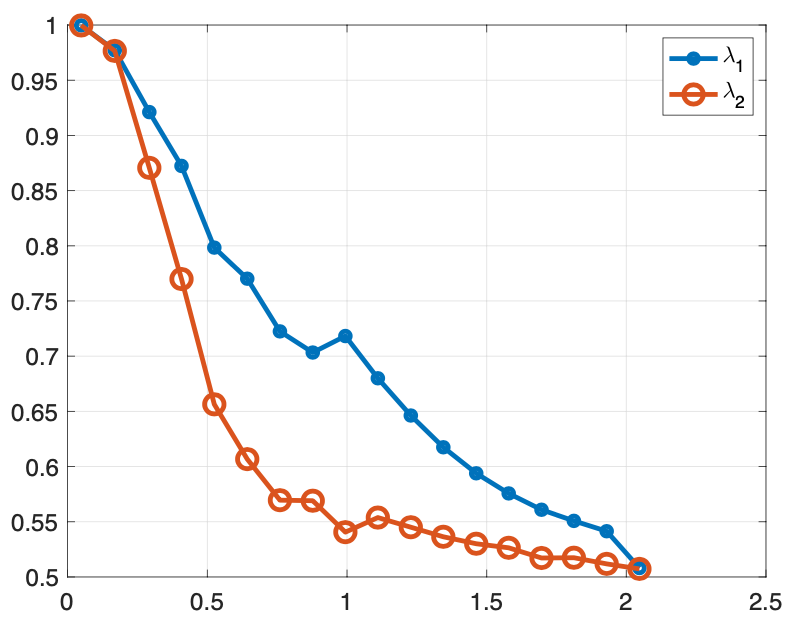}}\\
 \subfigure[SSIM vs NSR]{  \includegraphics[width=5cm]{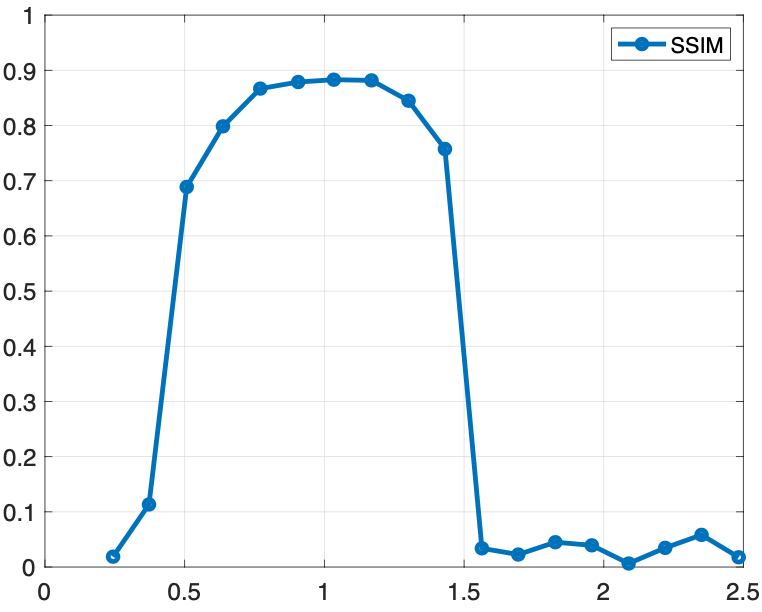}}\quad
\subfigure[SSIM vs NSR]{  \includegraphics[width=5cm]{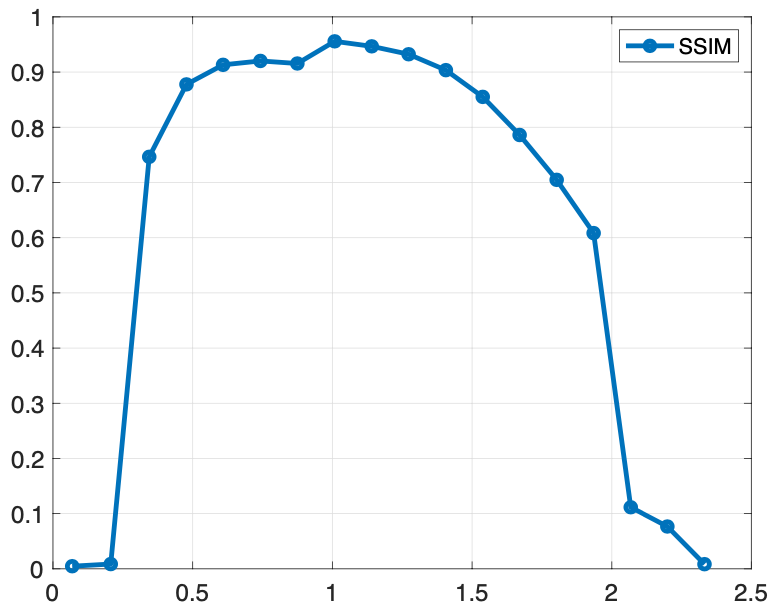}}\quad
 \subfigure[SSIM vs NSR]{  \includegraphics[width=5cm]{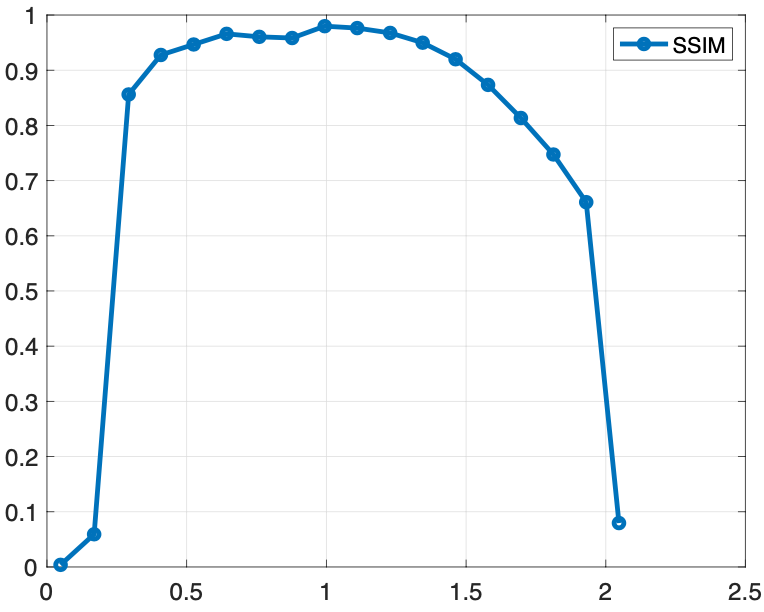}}
\caption{Correlations, eigenvalues and SSIM \eqref{ssim2} vs NSR with a fixed $(a)(d)(g) D=1.2 \times 10^6$; (b)(e)(h) $D=2.5\times 10^6$;  (c)(f)(i) $D=5\times 10^6$ for $D$ defined in \eqref{dose2}. }
\label{fig:opt-new}
\end{figure}

\subsection{Dose fractionation simulation}

For reconstruction, we test the shifted inverse power method at three levels of dose: $D=1.2\times 10^6, 2.5 \times 10^6, 5  \times 10^6$, according both definitions of $D$ discussed above.

Let $\lambda_{1}\!>\!\lambda_{2}$ be the two leading eigenvalues of $S^{-1}S_{\omega}$  (Sec. \ref{sec:1.1}). 
Figure \ref{fig:opt-new} shows the results with the definition \eqref{dose2} and Figure \ref{fig:opt} with the definition \eqref{dose}.
Both figures show  $R_1$ (resp. $R_2$),  the correlation between the true object and the leading (resp. the second leading) eigenvector, as NSR varies for a fixed  $D$. In both cases, the optimal dose fraction for all three levels of $D$ happens around NSR =1 where the spectral gap is also the largest.

Figure \ref{fig:opt-new} (g)(h)(i) also show the performance in terms of the structural component of SSIM which is very close to the quality metric $R_1$ because the pixel average is  close to zero for the object RPP and its reconstruction.

The main difference between Figure \ref{fig:opt} and Figure \ref{fig:opt-new} is a narrower $R_1$-curve and smaller spectral gap with the definition \eqref{dose} of $D$, due to the linear scaling behavior  \eqref{1st} for NSR greater than 1.

Figure \ref{fig:opt-new} and \ref{fig:opt}  clearly show  that the quality of reconstruction  
and the spectral gap are closely related to each other, both peaking at
$\mathrm{NSR}\simeq1$.  And, for both definitions,  \eqref{dose2} and \eqref{dose}, a lower level of $D$ gives rise to a narrower performance curve (in $R_1$, SSIM and spectral gap), making the optimal dose effect more pronounced.

\begin{figure}
 \centering   
 \subfigure[$R_1,R_2$ vs NSR]{  \includegraphics[width=5cm]{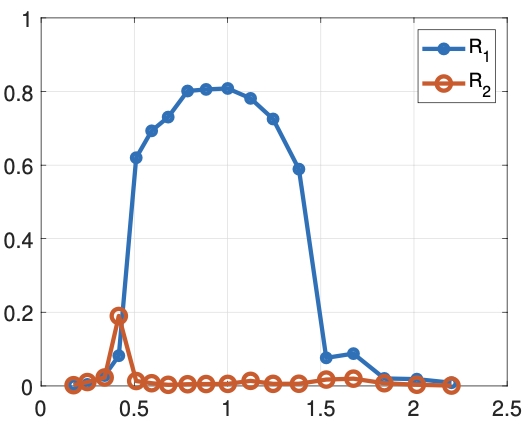}}\quad
 \subfigure[$R_1,R_2$ vs NSR]{  \includegraphics[width=5cm]{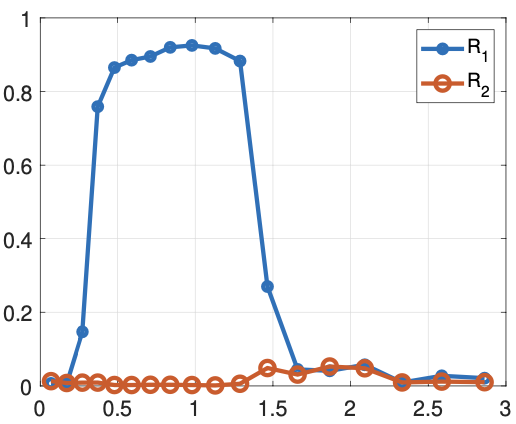}}\quad
 \subfigure[$R_1,R_2$ vs NSR]{  \includegraphics[width=5cm]{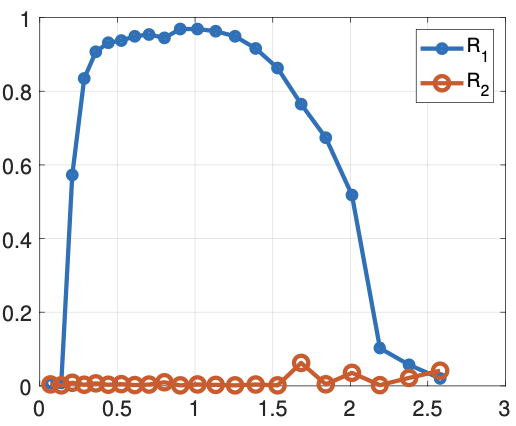}}\\
\subfigure[$\lamb_1, \lamb_2$ vs NSR]{  \includegraphics[width=5cm]{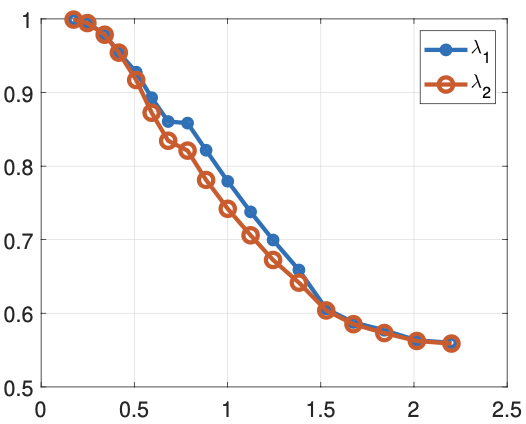}}\quad
\subfigure[$\lamb_1, \lamb_2$ vs NSR]{  \includegraphics[width=5cm]{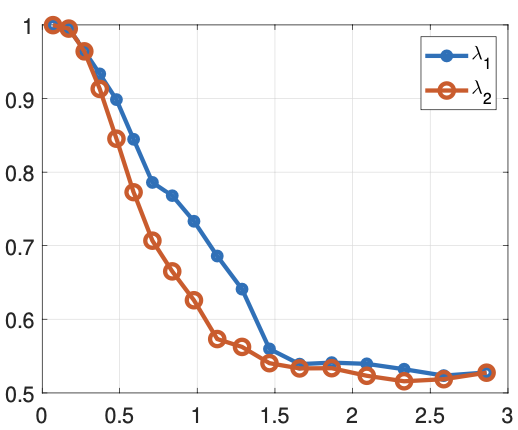}}\quad
\subfigure[$\lamb_1, \lamb_2$ vs NSR]{  \includegraphics[width=5cm]{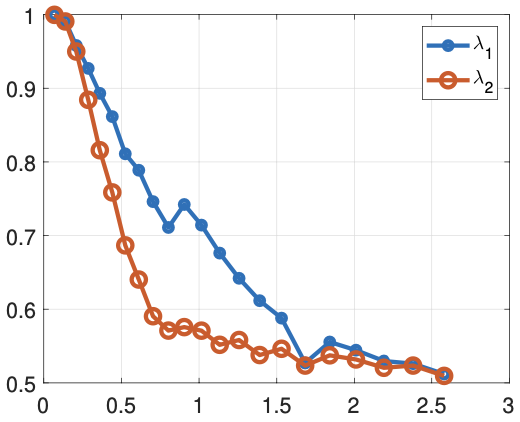}}
\caption{Correlations and eigenvalues vs NSR with a fixed $(a)(d) D=1.2 \times 10^6$; (b)(e) $D=2.5\times 10^6$;  (c)(f) $D=5\times 10^6$  for $D$ defined in \eqref{dose}.}
\label{fig:opt}
\end{figure}

\section{Conclusion and discussion}\label{sec:final}
We have  demonstrated the feasibility and robustness of 1-bit diffraction tomography in handling high-noise environments by using a random phase  mask.  

The optimal fractionation characterized by SNR = 1 means one photon per pixel on average is the informationally more efficient measurement. In view of the dose-resolution scaling
laws  $D\!\propto\!d^{-4}\sim n^4$ for coherent imaging \cite{dose-scaling09} and the number of pixels per diffraction pattern is $\sim n^2$, the number of diffraction patterns $m$ is then $\sim n^2$. This is fully consistent with the full‐sphere tomography (vs. few-axis tilt)  in,  e.g. single‐particle cryo‐EM or  X-ray diffractive imaging.

The surprising effect of  optimal dose fractionation near SNR =1 is not limited to 1-bit diffraction tomography. 
Similar effect persists in diffraction tomography with full intensity measurements (albeit in a less pronounced way). This is in part due to the 1-bit method is the state-of-the-art initialization method for iterative phase retrieval \cite{acta-phase}.

 Our conclusions regarding optimal fractionation (specifically, the empirical optimality at SNR = 1 per pixel) and the observed correlation between the spectral gap and eigenvector correlations are supported by numerical experiments rather than theoretical guarantees.  The existence and size of the spectral gap primarily govern the convergence and rate of the power iteration algorithm. By themselves, these quantities do not guarantee reconstruction accuracy or optimal dose allocation. 
These numerical findings point to a deeper underlying structure, but a rigorous analysis connecting spectral properties to image fidelity and dose efficiency remains an open and important direction for future work.

We have preliminarily shown that the proposed scheme performs equally well with 2-phase and 4-phase random masks. This has  potentially important practical implications such as easier fabrication at X-ray wavelengths, higher optical efficiency and faster/cheaper modulators for the optical spectrum \cite{kinoform}. 

 To what extent a continuous phase mask is replaceable by a quantized-phase mask requires further investigation. Randomness beating mask quantization depth holds only with a unit-variance phase mask, the maximum variance of any mask phase. This is a key to maintain  dose economy while reducing mask phase quantization. Notably  “dose economy” in imaging is not set by the total number of photons alone; it is set by how those photons are distributed over the resolution elements that carry information.
A random phase mask of less than unit variance such as the narrow-band binary phasor of
$\{\pm \pi/6\}$ of equal probability, while fabrication is easier -- shallower etch, higher throughput -- but information efficiency collapses, as many pixels get starved of counts while a few central pixels get most photons, degrading the quality of  reconstruction.

In this work, we have presented simulation results for a limited class of objects -- specifically, the ensemble of randomly phased phantoms. While incorporating additional prior information, such as real-valuedness, sparsity, or support constraints, is straightforward in principle -- e.g., by integrating alternating projection steps into the inverse power method -- the optimal strategy for doing so remains unclear. In particular, it is not yet known how best to balance these constraints with the spectral updates to achieve improved reconstruction accuracy and convergence robustness across broader object classes. Exploring these design choices is an important direction for future work.

Another limitation on the present work is the treatment of noise. Only the Poisson noise, which is fundamental in any microscopy, is considered here. The background and dark-current noises as well as the calibration errors are neglected. More important, we assume a {\em discrete} forward scattering model (Appendix A) and use it as the backward reverse model, bypassing  model-mismatch errors such as discretization errors and multiple scattering distortion.  An interesting question for future investigation  is how much of the optimal fractionation of dose at $NSR=1$ is attributable to the Poissonian nature of noise. 

Future directions could explore the incorporation of adaptive thresholding strategies, dynamic coding schemes, and hybrid measurement models to further enhance reconstruction quality and extend the applicability of this promising method.

\begin{appendix}

\section{Discrete tomography}\label{sec:discrete}

For the convenience of the reader and the completeness of presentation, we recall the discrete framework of diffraction tomography developed in \cite{Born-tomo,3D-phasing} and make connection to  the matrix setting of Theorem \ref{thm0}.

The discrete framework is a discretization of the continuum model often called the high-frequency Born approximation used in cryo-electron microscopy \cite{Frank06} and X-ray diffractive imaging \cite{Jacobsen}. See \cite{3D-phasing} for  the closely related  high-frequency Rytov approximation.  

In the continuum setting, the scattered exit waves modeled by the X-ray transform of $f$ undergo propagation and then are detected by detectors at far-field. 

 In a diffraction-limited imaging system with the wavelength $\lambda$,  the resolution length is roughly $\lambda/2$. Adopting a discrete framework for the X-ray transform, 
we set the grid spacing $\lambda/2$.  For simplicity, let  $\lambda=2$ so the grid spacing is 1. 
 
Let $\lb k,l\rb$ denote the integers
between and including the integers $k$ and $l$. 
Let  $O_n$ denote the class of discrete complex-valued objects
\beq
\label{1.1}
O_n:=\{f: f(i,j,k)\in \IC,  (i,j,k)\in \IZ^3_n; f(i,j,k)=0,  (i,j,k)\not\in \IZ^3_n\}
 \eeq
 where
\beq
\label{1.2}
\IZ_n&=&\left\{\begin{array}{lll}
\lb-n/2, n/2-1\rb && \mbox{if $n$ is an even integer;}\\
 \lb-(n-1)/2, (n-1)/2\rb && \mbox{if $n$ is an odd integer.}
 \end{array}\rt.
\eeq
To fix the idea, we consider the case of odd $n$ in the paper.

Following  the framework in \cite{discrete-X} we discretize the projection geometry as follows. 

We define three families of line segments, the $x_1$-lines, $x_2$-lines, and $x_3$-lines. 
The $x_1$-lines, denoted by $\ell_{(1,\alpha,\beta)}(c_1,c_2)$ with $ |\alpha|, |\beta|<1$,  are  defined by
\beq
\label{1.3'}
\ell_{(1,\alpha,\beta)}(c_1,c_2): \lt[\begin{matrix}
x_2\\
x_3\end{matrix}\rt]=\lt[\begin{matrix}\alpha x_1+c_1\\  \beta x_1+c_2
\end{matrix}\rt] && c_1, c_2\in \IZ_{2n-1}, \quad x_1\in \IZ_n
\eeq
 To avoid wraparound, we can zero-pad $f$ in a larger lattice $\IZ^3_p$ with $p\ge 2n-1.$ This is particularly important when it comes to define the ray transform by a line sum (cf. \eqref{2.8}-\eqref{2.10}) since wrap-around is unphysical. To fixed the idea, we define the object space that we shall work with:
 \beq
 \cX:= \{f\in O_n\;  \hbox{with the domain restricted to}  \:\IZ_{2n-1}^3\}.
 \eeq

Similarly, a $x_2$-line and a $x_3$-line are defined as
\beq
\ell_{(\alpha,1,\beta)}(c_1,c_2):  \lt[\begin{matrix}
x_1\\
x_3\end{matrix}\rt]=\lt[\begin{matrix}\alpha x_2+c_1\\  \beta x_2+c_2
\end{matrix}\rt]
 && c_1, c_2\in \IZ_{2n-1}, \quad x_2\in \IZ_n,\\
\ell_{(\alpha,\beta,1)}(c_1,c_2): 
 \lt[\begin{matrix}
x_1\\
x_2\end{matrix}\rt]=\lt[\begin{matrix}\alpha x_3+c_1\\  \beta x_3+c_2
\end{matrix}\rt] && c_1, c_2\in \IZ_{2n-1}, \quad x_3\in \IZ_n, 
\eeq
with $ |\alpha|, |\beta|<1$.

Let $\widetilde f$ be the continuous interpolation of $f$
given by
\beq\label{1.5}
  \widetilde f(x_1,x_2,x_3)&=&\sum_{i\in \IZ_n} \sum_{j\in \IZ_n}\sum_{k\in \IZ_n} f(i,j,k) D_p(x_1-i)D_p(x_2-j)D_p(x_3-k),\eeq
where 
 $D_p$ is the $p$-periodic Dirichlet kernel given by
\beq
D_p(t)={1\over p} \sum_{l\in \IZ_{p}} e^{\im 2\pi l t/p}
&=&\lt\{\begin{matrix}1,& t=mp,\quad m\in \IZ\\
{\sin{(\pi t)}\over p\sin{(\pi t/p)}},& \mbox{else}.
\end{matrix}\rt.\label{Dir}
\eeq
In particular, $[D_p(i-j)]_{i,j\in \IZ_p}$ is the $p\times p$ identity matrix. 
Because $D_p$ is a continuous $p$-periodic function, so is $\widetilde f$. 
However, we will only use the restriction of $\widetilde f$ to one period cell  $[-(p-1)/2, (p-1)/2]^3$ to define the discrete projections and  avoid the wraparound effect.

We define the discrete projections as the following line sums
\beq
 \label{2.8} 
 f_{(1,\alpha,\beta)}(c_1,c_2)&=& \sum _{i\in \IZ_n} \widetilde f(i,\alpha i+c_1,\beta i+c_2),\\
f_{(\alpha,1,\beta)}(c_1,c_2)&=& \sum _{j\in\IZ_n} \widetilde f (\alpha j+c_1,j, \beta j+c_2)\label{2.9}\\
f_{(\alpha,\beta,1)}(c_1,c_2)&=& \sum _{k\in \IZ_n} \widetilde f(\alpha k+c_1,\beta k+c_2, k)\label{2.10}
\eeq
with $ c_1, c_2\in \IZ_{2n-1}$.

The 3D discrete Fourier transform $\cF f$ of the object $f\in \cX$,  is given by
\beq
\cF f(\xi,\eta,\zeta)&=&p^{-3/2}\sum_{i,j,k} f(i,j,k)e^{-\im 2\pi(\xi i+\eta j+ \zeta k)/p}
\eeq
where
the range of the Fourier variables $\xi,\eta,\zeta$ can be extended from the discrete interval $\IZ_p$ to the continuum  $[-(p-1)/2, (p-1)/2]$. Note that by definition, $\widehat f$ is a $p$-periodic band-limited function. 
When there is no risk of confusion, we shall denote the {\em full} DFT for 
1D and 2D functions $g$ by  $\cF g$ and use the shorthand notation $\widehat g=\cF g$.

For $z$-lines, let $\bt=(\bt',1)$ with  $\bt'=(\alpha, \beta)$ denote the direction vectors.  Let $\mbf_\bt$ denote the discrete ray transform 
   \beq\label{fbeta}
\mbf_\bt(c):= \sum_{j \in \IZ_n}  \mbf \left(\bt'  j + c, j_3\right ),\quad \bc=(c_1, c_2)\in \IZ_p^2.
\eeq
Let $\cT$ denote the set of directions $\bt$ employed in the 3D diffraction measurement with a coded aperture (Figure \ref{fig1}). To fix the idea, {  let $p=2n-1$} in \eqref{Dir}.

Let $\mu$ be  the mask function  and $ f_\bt$ the object projection in the direction $\bt$. 

The Fraunhofer diffraction  of the masked scattered wave to the far-field detector  in Figure \ref{fig1}  is modeled by 
the  Fourier transform as
 \beq
\cF (\mbf_\bt \odot\mu)
&=&p^{-1}\widehat  \mbf_\bt \Conv \widehat  \mu(\bn):=p^{-1}\sum_{\bn'\in \IZ_p^2} \widehat  \mbf_\bt(\bn')\widehat  \mu(\bn-\bn')\label{28}
\eeq
with the resulting  coded diffraction patterns
\beq
\label{born-pattern}
\Big\{ |\cF (\mu\odot f_\bt) |^2,\quad\bt\in \cT\Big\}. 
 \eeq

 Given the randomness assumption  and  the asymptotic nature of the reconstruction accuracy guaranteed by Theorem \ref{thm0}, it is expected that 1-bit phase retrieval with coded diffraction patterns is similarly asymptotic in the sense that the quality  gradually increases as the numbers of diffraction patterns and random masks increase.

We  decompose the signal process 
  into two steps: 
  \beq\label{30}
 \cX\stackrel{\cR}{\longrightarrow} \cX_\cT:=\{(\mbf_\bt)_{\bt\in \cT}:  \mbf\in \cX\}\stackrel{\cQ}{\longrightarrow} 
 \{ { p^{-1} }(\widehat  \mbf_\bt\Conv \widehat \mu)_{\bt\in \cT}: \mbf\in \cX\}
  \eeq
   with  the ``collective" ray-transform $\cR$ and the masked 2D DFT $ \cQ$.
   Now we can write the  noiseless signal model
 as  $b=|\cA \mbf|$, with the measurement matrix $
 \cA:= \cQ\cR.$ The object domain projection  $\cP_\cX$ in Algorithm 1 is now carried out  in the transform domain  as   \bea\label{30'}
 \cP=\cA\cA^\dagger={\cQ\cR\cR^\dagger\cQ^*} \eea
 where $\cA^\dagger$ and $\cR^\dagger$ are the pseudo-inverses of $\cA$ and $\cR$ on $\cX$, respectively. 
    
For simplicity, let us assume that $\mu$ is a phase mask, i.e. $\mu=\exp[\im\phi], \phi\in \IR.$ Then for each $\bt$ the mapping \eqref{28} of $f_\bt$ is an isometry, modulo a scale factor, while 
the tomographic mapping $\cA$ defined for the 3D object $f$  is decisively not.  

To avoid  the missing cone problem in tomography, we consider
$m=3\rho n$  evenly distributed  random directions 
 \beq
\label{tset} \cT&=&\{ \bt_i=(1,\alpha_i,\beta_i) \}_{i=1}^{\rho n}\cup
\{ \bt_i=(\alpha_i,1, \beta_i) \}_{i=1+\rho n}^{2\rho n}\cup
\{ \bt_i=(\alpha_i, \beta_i,1) \}_{i=2\rho n+1}^{3\rho n}\\
&&\hbox{with $\alpha_i,\beta_i,i=1,\cdots,3\rho n,$ randomly chosen from
$(-1,1)$}\nn
\eeq
with the adjustable  parameter $\rho>0$.

 \section{Pseudo-inverse of $\cA$}  
 We adopt the following notation: $\cF_i, i=1,2,3,$ denotes the 1D DFT  in the $i$th variable over $\IZ_p$; $\cF_{ij},i,j=1,2,3,$ denotes the 2D DFT in the $i$- and $j$-th variables over $\IZ_p^2$; $\cF$ denotes the 3D DFT over $\IZ_p^3$. { Let $\cZ$ denote the zero-padding operator, $\cZ:\IC^{n\times n\times n}\to \IC^{p\times p\times p}$. Then 
 its adjoint  $\cZ^*$ is one projection, $\cZ^*:\IC^{p\times p\times p}\to \IC^{n\times n\times n}$.}
 
Since  $\cQ^*\cQ=I$, $A^\dagger =\cR^\dagger \cQ^*$  and hence the computation of $\cA^\dagger$  hinges on efficient implementation of $\cR^\dagger$. 

 First consider the case of projections along $z$-lines only.  
We have
   \beq
\mbf_\bt(\bc)&= &\sum_{j \in \IZ_n}  \mbf \left(\bt'   j + \bc, j\right )\nn\\
&=&p^{-2}\sum_{\bk}\sum_{j \in \IZ_n}\sum_{\bj} f(\bj,j)e^{\im2\pi (\bt' j+\bc-\bj)\cdot\bk/p}\nn\\
&=&p^{-2}\sum_{\bk} e^{\im 2\pi \bk \cdot \bc/p}\sum_{j \in \IZ_n} e^{\im2\pi  j\bt'\cdot\bk/p}
\sum_{\bj} f(\bj, j) e^{-2\pi \bk\cdot \bj/p}\nn \\
&=& p^{-1}\sum_{\bk} e^{\im 2\pi \bk \cdot \bc/p}\sum_{j\in \IZ_p} e^{\im2\pi  j\bt'\cdot\bk/p}\cF_{12}  (\cZ \mbf) (\bk, j).\label{54}
\eeq
For each $g=(g_j)\in \IC^p$, let $\Phi g $ be the quasi-periodic function given by 
  \beq \label{phi} \Phi g (\zeta)
   =\sum_{j\in \IZ^p} e^{\im2\pi  j\zeta/p} g_{j},\quad  g=(g_j)\in \IC^p. 
  \eeq
For given $\bk$ and $\cT=\{\bt_l=(\bt'_l, 1):l=1,...,m\}$ define the matrix $\Psi$ by
  \beq
  \cK g
  =\Big(\Phi g(\bt'_l\cdot \bk)\Big)_{l=1}^m,\quad g=(g_j)\in \IC^p
           \eeq
   whose adjoint is then given by
 \beq\label{eq_cK}
\cK^* h
&=&
\lt(\sum_{l=1}^m e^{-\i 2\pi  j\bt'_l \cdot    \bk/p} h_{l}\rt)_{j\in \IZ_p},\quad h:=(h_l)\in \IC^m. 
 \eeq

By \eqref{54} $\cR$ acting on $\cX$ can be decomposed as 
  \beq\label{eq103_}
{\cR=  \cF^*_{12}  \cK  \cF_{12}\cZ. }
\eeq

 Since
 \[
 \cR^\dagger=
( \cR^* \cR)^{-1} \cR^*, \; 
 \cR^* \cR=\cZ^* \cF^*_{12}  \cK^* \cK \cF_{12}\cZ
 \]
  the key to computing  $\cR^\dagger$ is the inversion of $\cR^*\cR$. Indeed, any
  \beq
  v:=\cR^\dagger h,\quad h\in\cX_\cT:=\{(\mbf_\bt)_{\bt\in \cT}:   \mbf\in \cX\}
  \eeq
satisfy  the normal equation
  \beq\label{eq_18}
 \cZ^* \cK^* \cK \cF_{12}\cZ  v&= &\cZ^* \cK^*  \cF_{12} h
 \eeq
which  is  a Toeplitz system in view of the identity  
\beq\label{eq_6_2}
\cK^* \cK  \cF_{12} \cZ v(\bk,j)= \sum_{i \in \IZ_p}  \sum_{l=1}^m e^{-\i 2\pi (j- i )\bt'_l \cdot    \bk/p}\cF_{12} \cZ v(\bk,i). \eeq

Observe that   $\cF_{12} \cZ v(\cdot,i)=0$ for all $i\not\in \IZ_n$ and $i-j\in \IZ_p$ for $i,j\in \IZ_n$  (with $p= 2n-1$). A Toeplitz matrix can be embedded into a circulant matrix, and the associated matrix-vector product can be implemented efficiently by FFTs\cite{Str86}\cite{CJ07
}. Hence  $ \cK^*\cK$ can be implemented  as  a $p\times p$  circulant  matrix when acting on vectors supported on $\IZ_n$:
\beq
\cK^*\cK&=& \lt(\sum_{l=1}^m  e^{-\i 2\pi   (j-i)\bt'_l  \cdot \bk /p}\rt)_{i,j\in \IZ_p}.\label{74}
\eeq
Consequently, $\cZ^* \cK^*\cK \cZ$ can be efficiently inverted  by diagonalizing with FFT. 
 This fast matrix-vector product  motivates the adoption  of conjugate gradient (CG) methods for computing $(\cR^* \cR)^{-1}$.

Indeed, let 
\beq 
w(\bk, k)&:=&\sum_{l=1}^m  e^{-\i 2\pi  k\bt'_l  \cdot \bk /p},\quad (\bk,k)\in \IZ_p^3. 
 \label{eq_u}
 \eeq
  By \eqref{eq_6_2} and the discrete convolution theorem
  \beq
 \cK^* \cK \cF_{12}\cZ v=w\ast_3 \cF_{12}\cZ v=\sqrt{p}
  \cF^*_{3} ( \cF_{3} w\odot \cF_{3} \cF_{12} \cZ v)
 \eeq
 where $\Conv_3$ denotes the discrete convolution over the third variable. 
Hence 
 \beq\label{eq_18_0}
   \cF^*_{12}   \cK^* \cK \cF_{12} \cZ v
 &=&  \sqrt{p} \cF^* (\cF_3 w \odot \cF \cZ v). 
  \eeq

 By  (\ref{eq_18}),  any $v=\cR^\dagger \mbh, h\in \cX_\cT$
 satisfy    \beq\label{CG_eq}
 \cR^* \cR v=\sqrt{p}\cZ^* \cF^*(\cF_3 w\odot \cF \cZ v)= \cZ^* \cF^*_{12}
 \cK^* \cF_{12} \mbh,\quad v\in \cX, 
 \eeq
 which is then solved by CG in the space $\cX$. 
 
In summary, we have $A^\dagger =\cR^\dagger \cQ^*$ where $\cR^\dagger$ is obtained by solving eq. \eqref{CG_eq} with $w$ given by \eqref{eq_u}. 

 \section{$\cA^\dagger$ with $x-$, $y-$ and $z-$lines}\label{general}
 Let 
\beq
 \cT_1&=&\{ \bt_i=(1,\alpha_i,\beta_i)\in \IR^3: i=1,\cdots, m_1 \}\\
\cT_2&=&
\{ \bt_i=(\alpha_i,1, \beta_i)\in \IR^3:  i=m_1+1,\cdots, m_1+m_2 \}\\
\cT_3&=&
\{ \bt_i=(\alpha_i, \beta_i,1)\in \IR^3:  i=m_1+m_2+1,\cdots, m_1+m_2+m_3 \}
\eeq
and $\cT=\cT_1\cup \cT_2\cup \cT_3$ be the total set of $m=m_1+m_2+m_3$ projection directions.

Let $\cA_j=\cQ_j\cR_j$ be the partial measurement matrix and decomposition corresponding to $\cT_j, j=1,2,3$. Let $\cA=[\cA^T_1,\cA^T_2,\cA^T_3]^T$ be the full measurement matrix which can be decomposed as $\cA=\cQ\cR$ where
$\cQ=\diag[\cQ_1,\cQ_2,\cQ_3]$ is the collective masked Fourier transform and $\cR=[\cR^T_1,\cR^T_2,\cR^T_3]^T$
the collective ray transform. 

To compute $\cA^\dagger=\cR^\dagger\cQ^\dagger$ the key is  $\cR^\dagger.$
Since $\cR^\dagger=(\cR^*\cR)^{-1}\cR^*$, we have
\beq
\label{83}
\cR^\dagger=(\cR_1^*\cR_1+\cR_2^*\cR_2+\cR_3^*\cR_3)^{-1} \cR^*.
\eeq
Analogous to \eqref{eq103_}, we have
\beqn
\cR_1&=&  \cF^*_{23}  \cK_1  \cF_{23}\cZ \\
\cR_2&=  &\cF^*_{13}  \cK_2  \cF_{13}\cZ\\
\cR_3&= & \cF^*_{12}  \cK_3  \cF_{12}\cZ. 
\eeqn
where for any $ g=(g_j)\in \IC^p$
\beqn
 \cK_1 g
  &=&\Big(\Phi g(\alpha_lk_2+\beta_lk_3)\Big)_{l=1}^{m_1}\\
 \cK_2g
  &=&\Big(\Phi g(\alpha_lk_1+\beta_lk_3)\Big)_{l=m_1+1}^{m_1+m_2}\\
  \cK_3 g
  &=&\Big(\Phi g(\alpha_lk_1+\beta_lk_2)\Big)_{l=m_1+m_2+1}^{m}
\eeqn
with the transform $\Phi$ defined by \eqref{phi}. Hence by \eqref{83}, the key to $\cR^\dagger$ is to invert  the operator
\beq
\cR^* \cR=\cZ^* (\cF^*_{23} \cK_1^* \cK_1  \cF_{23}+\cF^*_{13} \cK_2^* \cK_2 \cF_{13}+\cF^*_{12}  \cK_3^*\cK_3  \cF_{12})\cZ.
\eeq
Observe that from  the discrete convolution theorem, we have
 \beq\label{CG3_core}
\cR^* \cR v= \sqrt{p} \cZ^* \cF^*\Big(u\odot \cF \cZ v\Big),\quad u:=\sum_{l=1}^3\cF_l w_l, 
 \eeq
 where $\cR^*=[\overline{\cR_1},\overline{\cR_2},\overline{\cR_3}]$ (over-line denotes complex conjugation) and 
\beq 
w_1(k_1,k_2,k_3)&=&\sum_{l=1}^{m_1}  e^{-\i 2\pi  k_1(k_2\alpha_l+k_3\beta_l)/p},
 \label{eq_u1}\\
w_2(k_1,k_2,k_3)&=&\sum_{l=m_1+1}^{m_1+m_2}  e^{-\i 2\pi  k_2(k_1\alpha_l+k_3\beta_l)/p},
 \label{eq_u2}\\
 w_3(k_1,k_2,k_3)&=&\sum_{l=m_1+m_2+1}^{m}  e^{-\i 2\pi  k_3(k_1\alpha_l+k_2\beta_l)/p}. 
 \label{eq_u3}
 \eeq 
 Analogous to \eqref{CG_eq},
we  solve for $v=\cR^\dagger h$ from the equation:
 \beq\label{CG3_eq}
\sqrt{p} \cZ^* \cF^*\Big(u\odot \cF \cZ v\Big)= \cR^*\mbh. 
 \eeq

  \section{Preconditioners}
  We shall employ 
  preconditioned conjugate gradient methods (PCG) to 
   solve  (\ref{eq8}), where
  the preconditioner is   a
  circulant preconditioner related to   $\cR^* \cQ^* \cQ\cR=\cR^* \cR$.
  Suppose that all entries of $u$ in (\ref{CG3_core}) are positive.
 One natural preconditioner  
 is  Strang's circulant  preconditioner 
 $\cM$ satisfying 
\beq
\cM^{-1} v=  p^{-1/2} \cZ^* \cF^* \left\{u^{-1} \odot \cF( \cZ v)\right\},\; v\in \cX.
\eeq
Unfortunately, 
even though $\cR^*\cR$ is positive  definite, 
 these entries in $u$  vary widely, and some of them are negative. 
To alleviate the difficulty,  we introduce  a minor but necessary modification:
 \beq\label{Minv}
\cM^{-1} v=  p^{-1/2} \cZ^* \cF^* \left\{u^{-1} \odot \cF( \cZ v)\odot  (u \ge \epsilon)\right\},
\eeq 
where $\epsilon$ is chosen  to be a small positive scalar with $\cM\succ 0$.
 In  simulations,  $\epsilon=0.1$ is used.

\section{Krylov subspace methods to compute the second eigenvector} \label{sec:1.3}
For a large Hermitian matrix,
the Krylov subspace method is one efficient approach to compute eigenvectors of the extreme values.\cite{Saa96, GY02, GVL13}  
For instance, we can  employ the block Lanczos algorithm in 10.3.6 in \cite{GVL13} to compute a few extreme  eigenvectors of ${\cS }^{-1}\cS_{\omega} $.
The following  illustrates one algorithm that can efficiently compute the first two dominant eigenvectors of ${\cS }^{-1}\cS_{\omega} $. 

 Let $l$ be an integer with  $l\ge 2$. 
Let   $\mbf$ denote an $N\times l$ matrix,  consisting of  the  first $l$ dominant eigenvectors of $\cS^{-1}\cS_{\omega} $. Then 
 $\mbf$ can be regarded as a solution to 
 \beq\label{eq23}
\max_{\mbf}\left\{ \cG(\mbf):= \langle\mbf, \cS_{\omega} \mbf \rangle\right\}, \textrm{ subject to }\mbf^* \cS\mbf =I_l.
\eeq

Apply the Rayleigh–Ritz method to get    Ritz approximation of (\ref{eq23}). 
Construct a Krylov subspace spanned by the column space of some full rank matrix $X$. 
Consider the low dimensional approximate $\mbf$ in (\ref{eq23}), 
 $\mbf\approx X\beta$. Then 
 $\beta$ is a maximizer to 
 \beq\label{eq234'}
\max_{\beta}\langle X\beta, \cS_{\omega} X \beta\rangle
\textrm{
subject to 
} \beta^* X^* \cS X \beta=I_l.
\eeq
The following illustrates the determination of $\beta$.

\begin{prop} \label{prop1.1}Let $X$ be some $N \times k$ matrix with rank $k\ge l$. Introduce $C_{S_{\omega}}:=X^* \cS_{\omega}  X$,  $C_S:=X^* \cS X$. Suppose  $\mbf=X\beta$ for some $\beta \in \IC^{k\times l}$. Then $\beta$ satisfies  the maximization problem:
 \beq\label{eq234}
\max_{\beta}\langle X\beta, \cS_{\omega} X \beta\rangle=\max_{\beta}\langle \beta, C_{S_{\omega}}  \beta\rangle
\eeq
subject to 
\beq\label{eq235}
 \beta^* X^* \cS X \beta=\beta^* C_{S}  \beta=I_l.
\eeq
The optimal solution is 
  $\beta=C_{S} ^{-1/2} \alpha$, where  $\alpha$ is 
   the matrix whose columns are the first $l$ dominant eigenvectors of $ C_{S} ^{-1/2} C_{S_{\omega}}  C_{S} ^{-1/2}$.
\end{prop}

\begin{proof}
First, from (\ref{eq235}), we 
 express
  $\beta=C_{S} ^{-1/2} \alpha$ for some unitary  $k\times l$ matrix $\alpha$, i.e.,  $\alpha^* \alpha=I_l$. Second, from (\ref{eq234}),  $\alpha$ is a maximizer of 
\beq
\max_{\alpha }\langle C_{S} ^{-1/2}\alpha, C_{S_{\omega}}  C_{S} ^{-1/2}\alpha\rangle.
\eeq
Obviously, the optimal choice on $\alpha$ is the matrix whose columns are the first $l$ dominant eigenvectors of $ C_{S} ^{-1/2} C_{S_{\omega}}  C_{S} ^{-1/2}$.
\end{proof}

As noted in section~\ref{sec:1.1}, the shifted inverse iteration converges faster than the power iteration.
 We shall apply the operator  $(I-\cS^{-1}\cS_{\omega} )^{-1}$ to reach the invariant subspace of $\cS^{-1} \cS_{\omega} $, i.e., $span\{f_1,\ldots, f_l\}$.  This constructed subspace is known as the shift-and-invert Krylov subspace\cite{Eshof2006}, 
 which  was introduced 
 for the calculation of 
 the matrix exponential acting on a vector.
 
In summary, 
we repeat the two-step iterative procedure to reach a maximizer of $\cG(\mbf)$ in  (\ref{eq23}):\begin{itemize}
\item Let $\mbf=[f_1,\ldots, f_l]\in \IC^{n\times l}$. Form  $X\in \IC^{n\times lp}$, whose columns are  
 \beq
 \cup_{j=1,\ldots, l} \{
 f_j,  ((\cS-\cS_{\omega} )^{-1}\cS)f_j, \ldots, ((\cS-\cS_{\omega} )^{-1}\cS)^{p-1}f_j \}.\eeq

\item  Update  $f$ with  $\mbf=[f_1,\ldots, f_l]=X\beta$, where $\beta$ is determined from 
Prop.~\ref{prop1.1}.
\end{itemize}

\section{structural similarity index measure (SSIM)}\label{sec:ssim}
In its conventional formulation for real-valued images, the structural similarity index measure (SSIM) is defined in terms of sample mean $\mu$ and sample (co)-variance $\sigma$ as
\begin{equation}
\text{SSIM}(f,g) = \frac{2\mu_f \mu_g + c_1}{\mu_f^2 + \mu_g^2 + c_1}\cdot\frac{2\sigma_{fg}+c_2}{\sigma_f^2 + \sigma_g^2 + c_2},\label{ssim}
\end{equation}
where $c_1$ and $c_2$ are small constants to avoid instabilities.

For real‐valued images, SSIM compares local statistics—mean, variance, and covariance—to capture luminance, contrast, and structure differences and better correlates with human visual perception than MSE or PSNR.

For imaging problems, however, the luminance component (the first factor in \eqref{ssim}) is not as relevant as the other component. 
For simplicity, we will neglect the luminance component and set $c_2=0$, i.e. we will adopt the simplified expression
\beqn
\text{SSIM}(f,g)&=&\frac{2\sigma_{fg}}{\sigma_f^2 + \sigma_g^2}
\eeqn

For a complex valued object and reconstruction, we define the structural comparison function as
\beq\label{ssim2}
\text{SSIM}(f,g)&=&\frac{2|\sigma_{fg}|}{\sigma_f^2 + \sigma_g^2}, \quad 
\sigma_{fg} = \frac{1}{N-1} \sum_{i=1}^{N} \left( f_i - \mu_f \right) \left( g_i - \mu_g\right)^*
\eeq
where $N$ is the total number of pixels.

\end{appendix}

\bibliographystyle{alpha}
\bibliography{Tomographic}

\section*{Acknowledgement}
The research of PC is supported in part by grants
 110-2115-M-005 -007 -MY3
and
 111-2918-I-005 -002  from the National Science and Technology Council, Taiwan.
 The research of AF is supported in part by by the Simons Foundation grant FDN 2019-24 and the US National Science Foundation grant CCF-1934568. We thank to National Center for High-performance Computing (NCHC), Taiwan, for providing computational resources.

\end{document}